\documentclass[a4paper,onecolumn,superscriptaddress,11pt,accepted=2024-02-01]{quantumarticle}
\pdfoutput=1
\usepackage{dcolumn}
\usepackage{bm}

\usepackage{relsize}
\usepackage[utf8]{inputenc}
\usepackage[english]{babel}
\usepackage[T1]{fontenc}
\usepackage{hyperref}
\usepackage{graphicx}

  \usepackage[
    backend=biber,
    style=ieee,
  ]{biblatex}
\usepackage{amsthm}
\usepackage{amsmath} 
\usepackage{amssymb}
\usepackage{xcolor}
\usepackage[normalem]{ulem}
 
\newtheorem{theorem}{Theorem}
\newtheorem{corollary}[theorem]{Corollary}
\newtheorem{lemma}[theorem]{Lemma}
\newtheorem{definition}[theorem]{Definition}
\newtheorem{remark}[theorem]{Remark}
\newtheorem{example}[theorem]{Example}
\newtheorem{conjecture}[theorem]{Conjecture}

\begin{document}


\title{Stabilizer Codes with Exotic Local-dimensions}
\date{\today}

\author{Lane G. Gunderman}
\affiliation{No affiliation for this work}
\email{lanegunderman@gmail.com}


\date{February 5th, 2024}

\begin{abstract}
Traditional stabilizer codes operate over prime power local-dimensions. In this work we extend the stabilizer formalism using the local-dimension-invariant setting to import stabilizer codes from these standard local-dimensions to other cases. In particular, we show that any traditional stabilizer code can be used for analog continuous-variable codes, and consider restrictions in phase space and discretized phase space. This puts this framework on an equivalent footing as traditional stabilizer codes. Following this, using extensions of prior ideas, we show that a stabilizer code originally designed with a finite field local-dimension can be transformed into a code with the same $n$, $k$, and $d$ parameters for any integral domain. This is of theoretical interest and can be of use for systems whose local-dimension is better described by mathematical rings, which permits the use of traditional stabilizer codes for protecting their information as well.

\end{abstract}

\maketitle

The necessity of generating protected quantum information is a central challenge in developing large, controlled quantum computing devices. One of the earliest theoretical methods for correcting quantum information is the stabilizer formalism which encapsulated many of the existing example schemes and put them into a general framework \cite{gottesman1996class,gottesman1997stabilizer}. This framework permitted classical codes to be imported in certain cases, now known as the Calderbank-Shor-Steane (CSS) Theorem \cite{calderbank1996good,steane1996error}. While the stabilizer formalism does not encapsulate all error-correction schemes, it provides a particularly useful mathematical framework, and so this work will focus on quantum error-correcting codes that formally are stabilizer codes, but may best be described with atypical local-dimensions--such as infinite, continuous, and those with structure best described by mathematical rings. This work provides some useful uses for the \textit{local-dimension-invariant} (LDI) framework \cite{gunderman2020local,gunderman2022degenerate,moorthy2023local}. The final results in this work provide a general theorem somewhat similar to the CSS Theorem; whereas the CSS Theorem provides a method for bringing classical codes into quantum codes the results herein provide methods for using traditional quantum codes as quantum codes for other varieties of quantum systems.

There are two broad results that are focused on in this work. First, this work provides a solution to the long-standing open question of how to generally construct analog continuous-variable (CV) stabilizer codes \cite{lloyd1998analog,braunstein1998error}. These codes map $n$ physical continuous spaces into $k$ protected continuous spaces with the ability to recover from up to $\lfloor \frac{d}{2}\rfloor$ registers being damaged. Second, this work also generally provides methods for bringing stabilizer codes from any field to any selection of integral domain local-dimension. This setting is useful for modular-qudit stabilizer codes and some other novel settings \cite{ashikhmin2001nonbinary,gheorghiu2014standard,bullock2007qudit,ellison2022pauli,albert2020robust}. This latter result supersedes the former, but due to the independent utility of the prior we separate it out and discuss various cases related to it.


This work is organized as follows. First definitions are laid out, including the LDI framework, then we proceed to state our primary mathematical results. Following this, in Section \ref{results} we proceed to prove these mathematical results in various cases. We begin by showing the ability to use stabilizer codes for analog CV codes, then consider the two non-trivial approximate versions of this setting--where the continuous space is not infinite, and where the infinite space is not continuous. The section then proceeds to prove the general result of being able to promise the distance of the code over certain mathematical rings. Then in Section \ref{concl}, we discuss further possible extensions of the results to broader mathematical objects, including some discussion of topological codes which do not require the full distance promise of our proofs, as well as concluding remarks.

\section{Definitions and Background}

\subsection{Mathematical Definitions}

A mathematical, commutative field, in this work denoted by $\mathcal{F}$, is equipped with two invertible operations: addition and multiplication. A mathematical commutative ring \footnote{Throughout the work we use the term mathematical ring instead of simply ring in order to avoid any possible physical interpretable models that appear as rings, but do not have ring local-dimensions.}, in this work denoted by $\mathcal{R}$, is equipped with an invertible addition operation and a multiplication operation, but not all elements have multiplicative inverses. Then in this work $\mathcal{F}\subset \mathcal{R}$, and in particular we will require $\mathcal{F}$ to have a finite characteristic. As this work extends to cases beyond qubits as well as qudits, we will use the more general term \textit{registers} in place of these terms \cite{watrous2018theory}.

\subsection{Definitions for Stabilizer Codes}

In the discrete case with finite local-dimension $q$, the traditional qudit Pauli operators are generated by:
\begin{equation}
    X|j\rangle=|(j+1)\mod q\rangle,\quad Z|j\rangle=\omega_q^j |j\rangle,\quad \omega_q=e^{2\pi i/q}.
\end{equation}
We denote by $\mathbb{P}_q$ the single qudit Pauli operator group generated by the operators $X$ and $Z$, while $\mathbb{P}_q^n$ indicates a tensor product of $n$ such operators. In this work we will be varying the local-dimension, particularly focusing on the cases of formally infinite local-dimension, acting over $\mathbb{Z}$, and the continuous and simultaneously infinite local-dimension, acting over $\mathbb{R}$, cases. As our initial case we will consider $q$ as a prime number so that all nonzero values have a unique multiplicative inverse.

\begin{definition}
A stabilizer code, specified by its $n-k$ generators acting on $n$ registers, is characterized by the following set of parameters:
\begin{itemize}
\item $n$: the number of (physical) registers that are used to protect the information.
\item $k$: the number of encoded (logical) registers.
\item $d$: the distance of the code, given by the lowest weight of an undetectable generalized Pauli error. An undetectable generalized Pauli error is an $n$-qudit Pauli operator which commutes with all elements of the stabilizer group, but is not in the group itself.
\end{itemize}
These values are specified for a particular code as $[[n,k,d]]_q$, where $q$ is the local-dimension of the qudits.
\end{definition}


Working with tensors of Pauli operators can be challenging, and so we make use of the following well-known mapping from these to vectors in symplectic spaces, following the notation from \cite{gunderman2020local}. This representation is often times called the symplectic representation for the operators \cite{lidar2013quantum,ketkar2006nonbinary}, but we use this notation instead to allow for greater flexibility, particularly in specifying the local-dimension of the mapping. This linear algebraic representation will be used for our proofs.

\begin{definition}[$\phi$ representation of a qudit operator]
We define the linear surjective map: 
\begin{equation}
\phi_q: \mathbb{P}_q^n\mapsto \mathbb{Z}_q^{2n}
\end{equation}
which carries an $n$-qudit Pauli in $\mathbb{P}_q^n$ to a $2n$ vector mod $q$, where we define this mapping by:
\begin{equation}
I^{\otimes i-1} X^a Z^b I^{\otimes n-i} \mapsto \left( 0^{i-1}\ a\ 0^{n-i} \middle\vert 0^{i-1}\ b\ 0^{n-i}\right),
\end{equation}
which puts the power of the $i$-th $X$ operator in the $i$-th position and the power of the $i$-th $Z$ operator in the $(n+i)$-th position of the output vector. This mapping is defined as a homomorphism with: $\phi_q(s_1 s_2)=\phi_q(s_1)\oplus \phi_q(s_2)$, where $\oplus$ is component-wise addition mod $q$. We denote the first half of the vector as $\phi_{q,x}$ and the second half as $\phi_{q,z}$.
\end{definition}

Equivalently we can state this mapping as:
\begin{equation}
    \phi_q \left(\bigotimes_{t=1}^n X^{a_t}Z^{b_t}\right)=\left(\bigoplus_{t=1}^n a_t\right) \bigoplus \left(\bigoplus_{t=1}^n b_t\right),
\end{equation}
where in the above $\bigoplus$ is a direct sum symbol. We will write a vertical bar between the vector for the $X$ powers and the vector for the $Z$ powers mostly for ease of reading.

In the infinite case we may still define Pauli operators, denoted as $\mathbb{P}_\infty$, although with some slight adjustments to avoid it being trivially true in the limit. We define the $X$ and $Z$ operators in the natural way, however, the root of unity used must be somewhat carefully selected so that it has infinite multiplicative order:
\begin{equation}\label{infpaulis}
    X|j\rangle=|j+1\rangle,\quad X^{-1}|j\rangle=|j-1\rangle,\quad Z|j\rangle=\omega_\infty^j|j\rangle,\quad \omega_\infty:=e^{2\pi i /\sqrt{2}}
\end{equation}
Note that we selected $\sqrt{2}$ as the denominator in the phase, but any irrational number would suffice. This is since we wish for $\omega_\infty$ to have unbounded multiplicative order and an inverse. If we kept with the traditional definition for $\omega_q$ but took the limit of $q\rightarrow\infty$, formally we would have $\omega_\infty=1$, which would force $X$ and $Z$ to commute in the infinite case. In this case we have $XZ=\omega_\infty ZX$ and non-commutative properties still hold.

We may invert the map to return to the original $n$-qudit Pauli operator with the phase of each operator being undetermined. We make note of a special case of the $\phi$ representation:

\begin{definition}
Let $q$ be the dimension of the initial system. Then we denote by $\phi_\infty$ the mapping:
\begin{equation}
    \phi_\infty:  \mathbb{P}_\infty^n\mapsto \mathbb{Z}^{2n}
\end{equation}
where operations are no longer  taken modulo some base, but instead carried over the full set of integers.
\end{definition}

The ability to still define $\phi_\infty$ as a homomorphism, upon quotienting out the leading phase of the operators, (and with the same mapping) is a portion of the results of \cite{gunderman2020local}. Priorly this had only been defined so long as a new local-dimension was selected, however, in this work we formally show how this can be defined for codes with local-dimension given by the integers. $\phi_q$ is the standard choice for working over $q$ bases, however, our $\phi_\infty$ allows us to avoid being dependent on the local-dimension of our system when working with our code.

The commutator of two operators in this picture is given by the following definition:
\begin{definition}
Let $s_i,s_j$ be two qudit Pauli operators over $q$ bases, then these commute if and only if:
\begin{equation}
\phi_q(s_i)\odot \phi_q(s_j)\equiv 0 \pmod q
\end{equation}
where $\odot$ is the symplectic product, defined by:
\begin{equation}
\phi_q(s_i)\odot \phi_q(s_j) =\oplus_{k=1}^n [\phi_{q,z}(s_j)_k\cdot  \phi_{q,x}(s_i)_k- \phi_{q,x}(s_j)_k \cdot \phi_{q,z}(s_i)_k]
\end{equation}
where $\cdot$ is standard integer multiplication $\mod q$ and $\oplus$ is addition $\mod q$.
\end{definition}

Equivalently the commutator of two generalized Paulis, $p_1$ and $p_2$, is written and computed as $\phi(p_1)\odot \phi(p_2)=\vec{a}^{(1)}\cdot\vec{b}^{(2)}-\vec{b}^{(1)}\cdot\vec{a}^{(2)}$, with $\cdot$ as the standard dot product. This is not formally a commutator, but when this is zero, or zero modulo the local-dimension, the two operators commute, while otherwise it is a measure of the number of times an $X$ operator passed a $Z$ operator without a corresponding $Z$ operator passing an $X$ operator. When $q=2$, this becomes the standard commutation relations between qubit Pauli operators and is particularly simplified since addition and subtraction mod $2$ are identical.

In this work, we will assume that a stabilizer code is prepared in canonical form, given by $[I_{n-k}\ X_2\ |\ Z_1\ Z_2]$, where $X_2$ is a $(n-k)\times k$ block from putting the code into this form, while $Z_1$ is a $(n-k)\times (n-k)$ block and $Z_2$ is a $(n-k)\times k$ block also resulting from putting $\mathcal{S}$ into canonical form. So long as the local-dimension is a prime number, this can be done, although not uniquely per se. These form our basic tools for analyzing stabilizer codes.

\subsection{Local-dimension-invariant Codes}

Before delving into new results we will provide a brief summary of prior results for our primary tool of local-dimension-invariant stabilizer codes as introduced in \cite{gunderman2020local}. As an introductory example we may consider the qubit stabilizer code generated by the generators $\langle XX,ZZ\rangle$. Taking the symplectic product over the integers we find: $\phi_\infty(XX)\odot\phi_\infty(ZZ)=2$. This symplectic product means that as a qubit code these operators commute, but they fail to commute as generators for a qutrit code. If we instead used the generator $\langle XX^{-1},ZZ\rangle$, the symplectic product is now $0$, so these generators commute regardless of the local-dimension choice. In particular, for any choice of local-dimension $q$, this produces a single stabilized codeword given by:
\begin{equation}
    \frac{1}{\sqrt{q}}\sum_{j=0}^{q-1}|j,q-j\mod q\rangle,\quad \text{ with }\quad Z|j\rangle=e^{2\pi ij/q}|j\rangle
\end{equation}
This then inspires the following definition:
\begin{definition}
A stabilizer code $\mathcal{S}$ is called \textbf{local-dimension-invariant} iff:
\begin{equation}
    \phi_\infty(s_i)\odot \phi_\infty(s_j)=0,\quad \forall s_i, s_j\in \mathcal{S}
\end{equation}
\end{definition}

While a few examples happened to have been known prior to this line of inquiry, they were neither known by this name nor framework. This is likely why the prime examples were either the short 5-register code \cite{chau1997five} or the concatenated 9-register code \cite{chau1997correcting}. These examples, along with the above and other simpler codes, are not the only local-dimension-invariant codes. The following theorem was shown in \cite{gunderman2020local}, but for completeness we provide it here as well since we will need to modify it later. This Theorem provides a prescriptive method for transforming any stabilizer code into an LDI form:
\begin{theorem}\label{ogldi}
Any stabilizer code $\mathcal{S}$, with parameters $[[n,k,d]]_q$ for prime $q$, can be put into a local-dimension-invariant form.
\end{theorem}

\begin{proof}
We provide a constructive method for transforming a given stabilizer code $\mathcal{S}$ into a local-dimension-invariant form $\mathcal{S}'$. The primitives for a finite field Gaussian elimination is obtained through the composition of generators in the $\phi$ representation, swapping rows, applying discrete Fourier transform gates (to swap $X$ and $Z$ powers, up to sign), and swapping registers. By performing this Gaussian elimination we put the code into canonical form: $[I_{n-k}\ X_2\ |\ Z_1\ Z_2]$, where $X_2$ is a $(n-k)\times k$ block from putting the code into this form, while $Z_1$ is a $(n-k)\times (n-k)$ block and $Z_2$ is a $(n-k)\times k$ block also resulting from putting $\mathcal{S}$ into canonical form.

Now consider the modified set of generators given by $\mathcal{S}':=[I_{n-k}\ X_2\ |\ Z_1+L\ \ Z_2]$, with $L$ the lower triangular matrix with $L_{ij}=\phi_\infty(s_i)\odot\phi_\infty(s_j)$ for $i>j$, and $0$ otherwise. $\mathcal{S}'$ exhibits two properties:
\begin{enumerate}
    \item $\mathcal{S}'\equiv\mathcal{S}\ (\mod q)$, meaning that the code is unaltered over the original local-dimension choice.
    \item For all $s_i',s_j'\in \mathcal{S}'$, we have:
    \begin{eqnarray}
        \phi_\infty(s_i')\odot \phi_\infty(s_j')&=&(\phi_\infty(s_i)+(0\ |\ L_i))\odot (\phi_\infty(s_j)+(0\ |\ L_j))\\
        &=&\phi_\infty(s_i)\odot \phi_\infty(s_j)-L_i\cdot \phi_\infty(s_j)_x+L_j\cdot \phi_\infty(s_i)_x\\
        &=& \phi_\infty(s_i)\odot \phi_\infty(s_j)-L_{ij}+L_{ji}\\
        &=& 0,
    \end{eqnarray}
    where the last step comes from $L$ being lower-triangular.
\end{enumerate}
Therefore $\mathcal{S}'$ is in LDI representation for $\mathcal{S}$ and still has $n-k$ generators.
\end{proof}

This is not the only way to form an LDI representation. Some other methods are briefly discussed in \cite{gunderman2022degenerate}. For the sake of simplicity we will use the lower-triangular $L$ matrix method when using a prescriptive method. Throughout this work we will feature the example of the Steane code as this code is not a priori local-dimension-invariant \cite{steane1996multiple}.

Permitting a slight abuse of the $\phi$ notation, an LDI form for the Steane code, which in Pauli form we denote by $H_3$ due to its correspondence with the classical Hamming code, is given by the following \cite{moorthy2023local}:
\begin{equation}\label{ldisteane}
\phi_\infty(H_3)=
\setcounter{MaxMatrixCols}{15}
\begin{bmatrix}
 1 & 1 & 1 & 1 & 0 & 0 & 0 & | & 0 & 0 & 0 & 0 & 0 & 0 & 0\\
 0 & 1 & 1 & 0 & 1 & 1 & 0 & | & 0 & 0 & 0 & 0 & 0 & 0 & 0\\
 0 & 0 & 1 & 1 & 0 & 1 & 1 & | & 0 & 0 & 0 & 0 & 0 & 0 & 0\\
 0 & 0 & 0 & 0 & 0 & 0 & 0 & | & 1 & -1 & 1 & -1 & 0 & 0 & 0\\
 0 & 0 & 0 & 0 & 0 & 0 & 0 & | & 0 & 1 & -1 & 0 & 1 & -1 & 0\\
 0 & 0 & 0 & 0 & 0 & 0 & 0 & | & 0 & 0 & 1 & -1 & 0 & -1 & 1\\
\end{bmatrix}.
\end{equation}
We will work with this LDI representation since it is also in CSS form and only differs from the traditional Steane code in signs.


$\phi_\infty(H_3)$ represents a set of generators that each have pairwise symplectic product zero, meaning that these generators commute regardless of the local-dimension of the underlying system.

The prior only ensures that there are generators over these local-dimensions. In order to ensure that the distance of the quantum error-correcting code of the LDI representation is at least as large as that of the original stabilizer code, we reproduce the theorem from \cite{gunderman2020local}:
\begin{theorem}\label{ogproof}
For all primes $p>p^*$, with $p^*$ a cutoff value greater than $q$, the distance of an LDI form for a stabilizer code $[[n,k,d]]_q$ used for a system with $p$ bases, $[[n,k,d']]_p$, has $d'\geq d$.
\end{theorem}

To make claims about the distance of the code we begin by breaking down the set of undetectable errors into two sets. These definitions highlight the subtle possibility of the distance reducing upon changing the local-dimension.

\begin{definition}\label{unavoidable}
An \textbf{unavoidable error} is an error that commutes with all stabilizers and produces the $\vec{0}$ syndrome over the integers.
\end{definition}

These correspond to undetectable errors that would remain undetectable regardless of the local-dimension for the code since they always exactly commute under the symplectic inner product with all stabilizer generators--and so all members of the stabilizer group. Since these errors are always undetectable we call them unavoidable errors since changing the local-dimension would not allow this code to detect this error.

We also define the other possible kind of undetectable error for a given number of bases, which corresponds to the case where some syndromes are multiples of the local-dimension:

\begin{definition}\label{artifact}
An \textbf{artifact error} is an error that commutes with all stabilizers but produces at least one syndrome that is only zero modulo the base.
\end{definition}

These are named artifact errors as their undetectability is an artifact of the local-dimension selected and could become detectable if a different local-dimension were used with this code. Generally there is no relation between artifact errors for different local-dimension choices. As we will see, however, for very large local-dimensions this set becomes empty. Each undetectable error is either an unavoidable error or an artifact error. We utilize this fact to show our theorem.

\begin{proof}
The ordering of the stabilizers and the ordering of the registers does not alter the distance of the code. With this, $\phi_\infty$ for the stabilizer generators over the integers can have the rows and columns arbitrarily swapped. Before proceeding, let us flip the signs of all entries in the $Z$ component so that the symplectic product can be treated as a regular inner product, and symplectic kernels are replaced by traditional kernels of vector spaces.

Let us begin with a code over $q$ bases and extend it to $p$ bases. The errors for the original code are the vectors in the kernel of $\phi_q$ for the code. These errors are either unavoidable errors (including the stabilizer group itself) or are artifact errors. We may rearrange the rows and columns so that the stabilizers and registers that generate these syndrome values that are nonzero multiples of $q$ are the upper left $2d\times 2d$ minor, padding with identities (zero vectors in the $\phi$ representation) if needed. The factor of 2 occurs due to the number of nonzero entries in $\phi_\infty$ being up to double the weight of the Pauli. The stabilizer(s) that generate these multiples of $q$ entries in the syndrome, upon taking the product with an error, form the rows for a matrix minor for which there is a null space, and for which the number of columns used are given by the weight of the Pauli operator. In essence, we form a matrix minor from entries in the stabilizers in the support of the error, and suppose that such a matrix has a null space. This null space corresponds to artifact errors.

Now, consider the extension of the code to $p$ bases. Building up the errors by Pauli weight, considering weight $j$, we consider the minors of the stabilizer matrix composed through all row and column swaps. This is an over complete method as this generates not only all minors with support corresponding to errors but any possible minor, however, it suffices. These minors of size $2j\times 2j$ can have a nontrivial null space in two possible ways:
\begin{itemize}
    \item If the determinant is 0 over the integers then this is either an unavoidable error or an error whose existence did not occur due to the choice of the local-dimension.
    \item If the determinant is not 0 over the integers, but takes the value of some multiple of $p$, then it's $0\mod p$ and so a null space exists.
\end{itemize}
Thus we can only introduce artifact errors to decrease the distance. By bounding the determinant by $p^*$, any choice of $p>p^*$ will ensure that the determinant is a unit in $\mathbb{Z}_p$, and hence has a trivial null space since the matrix is invertible.

Now, in order to guarantee that the value of $p$ is at least as large as the determinant, we can use Hadamard's inequality to obtain:
\begin{equation}\label{boundone}
    p> p^* =B^{2(d-1)}(2(d-1))^{(d-1)}
\end{equation}
where $B$ is the maximal entry in $\phi_\infty$ \footnote{Maximal here is taken as the largest absolute value over the reals. Depending on the method used for generating the LDI representation this can be bounded above by $\max_{ij}|\phi_\infty(s_i)\odot\phi_\infty(s_j)|$ or $(2+k(q-1))(q-1)$.}. Since we only need to ensure that the artifact induced null space is trivial for Pauli operators with weight less than $d$, we used this identity with $2(d-1)\times 2(d-1)$ matrices.

When $j=d$, we can either encounter an unavoidable error, in which case the distance of the code is $d$ or we could obtain an artifact error, also causing the distance to be $d$. It is possible that neither of these occur at $j=d$, in which case the distance becomes some $d'$ with $d<d'$. 
\end{proof}

While the distance may increase if all weight $d$ errors are artifact errors and are avoided through the local-dimension change, it cannot do so indefinitely.

\begin{definition}\label{dstar}
The distance $d^*$ of a stabilizer code is given by the minimal weight unavoidable error.
\end{definition}
The value of $d^*$ for a code effectively provides a limit to the distance of the quantum error-correcting code that cannot be overcome by increasing the local-dimension of the code. In the infinite local-dimension case, with or without continuous bases, we will always achieve this $d^*$ value, as we will show. A large portion of the remaining results in this work are effectively variations of this proof with different settings, as well as some expanded settings for the results.

\subsection{Statement of Results}

As a number of nuances in definitions must be made and the proof of these results may not be of particular interest to all readers, we separate out the primary results here, aside from this first Theorem which shows that if the local-dimension is a field then the results still hold. Being restricted to finite fields requires that the local-dimension just be a prime-power, but if one wishes to use stabilizer codes for other local-dimensions then we need to expand beyond fields to mathematical rings. The primary challenge of having mathematical ring local-dimensions is that the powers of the Pauli operators no longer are guaranteed to have multiplicative inverses. This next result generally provides that $n$ and the number of linearly independent generators can be kept even if mathematical rings are the correct description for the local-dimension of a system. Unfortunately, since addition of generators in the $\phi$ representation corresponds to composition of generators, being linearly independent does not always suffice for being able to additively generate the full space of syndrome detection required. The common case of this is that of qudits with $q=p^m$, $m>1$, whereby the syndromes that can be obtained are from $\mathbb{Z}_p$, so including each generator times all non-trivial powers of an irreducible factor permits the full error-correction capabilities \cite{gottesman2014stabilizer}. Although other cases can occur, we will use this observation to obtain statements on the number of logical registers of the LDI code \cite{nadkarni2021mathbb,garani2023theory}. In the case of finite, commutative mathematical ring local-dimensions, we may decompose the ring in a similar way such that each generator multiplies all non-trivial power of each irreducible polynomial--for instance for codes with local-dimension matching the structure of $\mathbb{Z}_6$ we would use all powers of the irreducible polynomials with order $2$ and $3$. Using the Chinese remainder theorem ring homomorphism then decomposes the $\mathbb{Z}_6$ local-dimension into $\mathbb{Z}_2\oplus \mathbb{Z}_3$, which have characteristics $2$ and $3$, respectively.

\begin{theorem}\label{fldi}
Let $\mathcal{S}$ be a stabilizer code with parameters $[[n,k]]_{\mathcal{F}}$ for local-dimension given by finite field $\mathcal{F}$, with $r$ linearly independent generators in the $\phi$ representation. Then there is a prescriptive method to transform $\mathcal{S}$ into a local-dimension-invariant form still with parameters $[[n,n-r]]_{\mathcal{R}}$ for any choice of a finite, commutative ring $\mathcal{R}$.
\end{theorem}

\begin{proof}[Proof of Theorem \ref{fldi}]
As in the initial proof of the canonical form \cite{gottesman1997stabilizer}\footnote{Whilst this initial proof did not form precisely this canonical form, it generates a form that differs by a discrete Fourier transform and some row additions, so we provide credit there.}, composition of generators corresponds to row additions, multiplying by a non-zero element of $\mathcal{F}$ generates an equivalent generator, rows may be swapped freely, as well as registers, and lastly $X$ powers and $Z$ powers for a register may be swapped (up to sign) by a discrete Fourier transform \cite{gottesman1999fault}, generating a full set of elementary matrix operations. While in the prime power case formally only additivity is permitted, momentarily using linearity, finding an LDI form, then undoing the linearity used is permitted since the linearity would only be used to determine the entry changes required. Let there be $r$ linearly independent generators. Using this we may perform Gaussian elimination to reduce the generators to canonical form: $[I_r\ X_2\ |\ Z_1\ Z_2]$.

The same lower-triangular $L$ matrix as in the proof of Theorem \ref{ogldi} may be added to the $Z_1$ block to generate commuting generators which leave $\mathcal{S}$ alone over $\mathcal{F}$. Lastly, to preserve the independence number, $r$, when the vector space over $\mathcal{F}$ spanned by the generators is transformed into a module over a ring $\mathcal{R}$, it suffices to note that the first $r$ entries are still independent for any ring $\mathcal{R}$, so all generators are still independent.

In the case of rings, since each generator, $g$, has syndrome values in $\text{char}(g)$, we generally must augment the generators from the canonical form in order to create an additively spanning set. For a finite ring $\mathcal{R}$ the characteristic of any element divides $|\mathcal{R}|$, so including cosets of the generator $g$ multiplied by powers of an irreducible factor for each characteristic factor, will create an additively spanning module. Then each generator will be replaced by $|\mathcal{R}|$ generators, which will result in $n-r$ logical registers--and thus an $[[n,n-r]]_{\mathcal{R}}$ quantum error-correcting code.
\end{proof}


First, it is worth noting that this reduces to the prior LDI construction theorem, here Theorem \ref{ogldi}. It is also worth noting that this does not preclude stabilizer codes over broader mathematical objects from also being able to be put into an LDI representation. A particularly notable instance are codes over the integers, $\mathbb{Z}$, where the mathematical ring is infinite, however, the linearly independent generators still form an additively spanning set. The central element needed for other mathematical objects is the ability to put the code into canonical form, so even if multiplicative inverses do not exist for the initial local-dimension, if the code is in canonical form we may still put it into an LDI representation and ensure the rank. A possible method to include composite local-dimensions as well could be to use a Smith normal form such as that in \cite{sarkar2023qudit}. 




The following result provides a plethora of stabilizer codes for the analog CV quantum error-correcting scheme laid out in \cite{lloyd1998analog,braunstein1998error}. It broadly says that any stabilizer code can be used as an analog CV quantum error-correcting code and will have the same number of registers and logical registers, and at least the same distance.

\begin{theorem}[CV Analog Case]
Let $q$ be a prime number and $\mathbb{R}$ be the real numbers. Then any stabilizer code $[[n,k,d]]_q$ can be transformed into a CV analog code with parameters $[[n,k,d']]_{\mathbb{R}}$ with $d'=d^*\geq d$.
\end{theorem}

This result is shown in \ref{cvthm}, with $d^*$ as in Definition \ref{dstar}. Prior to this result, there were relatively few known non-trivial examples of analog CV codes \cite{barnes2004stabilizer,albert2022bosonic}, so this provides a tremendous choice of codes. While the above Theorem is stated as requiring the initial local-dimension to be a prime number, this is not strictly required and any mathematical field $\mathcal{F}$ can be the initial local-dimension instead. The prior theorem is formally just a special case of the following theorem:



\begin{theorem}[$\mathcal{F}\mapsto \mathcal{R}$ Case]
Let $\mathcal{F}$ be a finite field and $\mathcal{R}$ be an integral domain. Then if $\text{char}(\mathcal{R})=0$ or $\text{char}(\mathcal{R})>p^*$, with $p^*$ some cutoff value that is a function of $\text{char}(\mathcal{F})$, $n$, $k$, and $d$, a stabilizer code $[[n,k,d]]_{\mathcal{F}}$, with $r$ linearly independent generators in the $\phi$ representation, can be used as a $[[n,n-r,d']]_{\mathcal{R}}$ stabilizer code with $d'\geq d$.
\end{theorem}

This Theorem, shown in \ref{ringthm}, covers the various restrictions of the CV Analog code case: continuous-variable but bounded, integer local-dimension, and qudits. Additionally, it permits for codes without physical realizations but may be of theoretical interest which are best described by integral domains. Unfortunately, composite local-dimensions are not able to be shown here and so are left as a future direction. Lastly, through a synthesis of cases, rotor codes and their parameters are able to be obtained.

A single example family of this Theorem's use is to transform the $[[2^N-1,2^N-1-2N,3]]_2$ stabilizer code into a $[[2^N-1,2^N-1-2N,3]]_{\mathcal{R}}$ stabilizer code for any ring $\mathcal{R}$ with characteristic greater than $2$, an example that will be delved into slightly greater depth later along with other examples. Having stated our achieved results we now proceed to prove them and provide some settings for their possible uses.

\section{Local-dimensions: the continuous, the infinite, and the exotic cases}\label{results}

While the initial version of stabilizer codes operated on qubits, this was extended to the qudit case in time, typically under the constraint that the local-dimension is still a prime or prime-power. In this section we show that the same formalism can be extended to a number of broader mathematical settings. In particular, we begin by showing how using the LDI framework one may import stabilizer codes for use as analog CV error-correcting codes. Following this we specialize to the case whereby the underlying phase space is able to be interpreted as being a repeating cell, which is equivalent to having quadrature operators upper bounded by some value. In this case the distance of the generated code is only promised so long as the upper bound is sufficiently large. This reflects a somewhat more realistic phase space model for analog CV codes as they will have finite extent in reality. After this specialization, we consider the other approximation to analog CV codes whereby the local-dimension is allowed to be infinite, but only taking discrete values--in essence codes over $\mathbb{Z}$. We find that codes with this constraint are rather reminiscent of GKP encodings and form lattice points for their codewords which must be tapered to be physically realizable. The restriction to both finite cell and discrete values returns the codes to the traditional qudit case, so we do not examine those here.

Following these various cases of analog CV codes, we proceed to consider an extension of the stabilizer formalism using the LDI framework and the insights gained from considering codes over $\mathbb{Z}$. Using a similar proof technique we are able to show the ability to utilize stabilizer codes which originate over a finite field local-dimension to systems with integral domain local-dimensions, and give a condition under which the distance of such codes will be at least as large as the original stabilizer code. This scenario will primarily appeal to quantum error-correcting code theorists, but could be of use for topological systems exhibiting local-dimensions which are best described by mathematical rings.

Throughout the sections we provide some examples to illustrate the results and attempt to contextualize the results within some possible physical settings.


\subsection{Analog CV QECC: Continuous and Infinite}\label{cvthm}



In this section we show how to import arbitrary stabilizer codes, both qubit and qudit, and use them for analog CV codes, in the style of Lloyd and Slotine \cite{lloyd1998analog} and Braunstein \cite{braunstein1998error}. In the continuous-variable setting the Heisenberg group, displacements of the quadratures, defines the Pauli operators and the errors. To begin, we define the "additive" form for Pauli operators which replaces a tensor product of many Pauli operators with a sum of single quadrature terms. Mathematically the quadrature operators, sometimes labeled as $x$ and $p$, satisfy the commutation relation $[x,p]=1$. Physically these may be position and momentum, electric and magnetic fields, or other dual variables and the domain of possible values for these quadratures form the phase space. A vector of these quadratures, representing a set of modes, is specified by $\vec{x}$ and $\vec{p}$. These quadrature sums form nullifiers which generate the code space \cite{albert2022bosonic}. The following Lemma shows this mapping and proves that Pauli operators represented this way observe the same commutation conditions \footnote{This result was independently shown in \cite{barnes2004stabilizer}, but we show it here in our language.}.

\begin{lemma}[Product to additive form]\label{addform}
Let $\mathcal{S}$ be a stabilizer code, with $\phi_\infty(\mathcal{S})$ being the integer symplectic representation. Then there is a commutation conserving mapping such that:
\begin{equation}
    s_i:=X^{\vec{a}} Z^{\vec{b}}\mapsto \vec{a}\cdot \vec{x}+ \vec{b}\cdot \vec{p}:=\mathcal{A}[s_i]
\end{equation}
with $\mathcal{A}$ being the additive form of the analog nullifier.
\end{lemma}

\begin{proof}
To show this, we must show that $\phi_\infty(s_i)\odot \phi_\infty(s_j)=0$ iff $[\mathcal{A}[s_i],\mathcal{A}[s_j]]=0$. First, the symplectic side states that:
\begin{equation}
    \vec{a}_i\cdot \vec{b}_j-\vec{a}_j\cdot \vec{b}_i=0.
\end{equation}
The additive form can be written as:
\begin{eqnarray}
    [\mathcal{A}[s_i],\mathcal{A}[s_j]]&=& [\vec{a}_i\cdot \vec{x}+ \vec{b}_i\cdot \vec{p},\vec{a}_j\cdot \vec{x}+ \vec{b}_j\cdot \vec{p}]\\
    &=& [\vec{a}_i\cdot \vec{x}, \vec{b}_j\cdot \vec{p}]+[\vec{b}_i\cdot \vec{p}, \vec{a}_j\cdot \vec{x}]\\
    &=& \sum_k a_{ik}b_{jk} [x_k,p_k]-b_{ik}a_{jk}[x_k,p_k]\\
    &=& \vec{a}_i\cdot \vec{b}_j-\vec{a}_j\cdot \vec{b}_i,
\end{eqnarray}
and so the two are equal and so if the operators commute in one representation they will in the other.
\end{proof}

With this mapping from Pauli to additive form, we may still represent the same stabilizer code with the same symplectic matrix, only noting that now instead of the entries representing powers of Pauli operators they now represent coefficients of quadrature operators. This now permits us to show the following theorem:

\begin{theorem}\label{cv}
Let $\mathcal{S}$ be a stabilizer code with parameters $[[n,k,d]]_q$, for a prime $q$, then $\mathcal{S}$ can be used as an analog CV quantum error-correcting code with parameters $[[n,k,d']]_{\mathbb{R}}$, with $d'=d^*\geq d$.
\end{theorem}

The conditions for transforming $\mathcal{S}$ into a valid set of nullifiers for an analog CV error-correcting code are provided above. In order for the nullifiers to commute, all that must be done is to transform a given stabilizer code into an LDI representation then turn it into an "additive" form. These, however, only provide for ensuring that the analog CV code uses $n$ physical continuous modes to generate $k$ logical continuous modes. To show that the distance of the transformed code is at least as good requires some observations. Recall Definitions \ref{unavoidable} and \ref{artifact}, which broke the set of undetectable errors into different varieties, for this proof.

\begin{proof}
The important insight here is that the local-dimension in this case is effectively infinite, and so artifact errors inherently cannot be induced. Digging into this a little more, the traditional argument for preserving the distance of an LDI form for a code involves ensuring that one cannot generate any new kernels for the minors of size up to $2(d-1)\times 2(d-1)$ in the LDI representation. This is guaranteed by bounding the determinant using Hadamard's inequality, providing requirements on how large the local-dimension must be to avoid incidentally introducing an artifact error which reduces the distance of the code. In this case, assuming infinite precision and infinitely large coefficients, there is no way to introduce an artifact error, and so the distance of the analog code must be some $d'$ with $d'\geq d$. In fact, the distance of these codes is $d^*$, the minimal weight unavoidable error, since only unavoidable errors determine the distance of the code generated.
\end{proof}

The LDI framework provided for a relatively easy proof of this Theorem. This means that, were one to apply this with recent good quantum Low-Density Parity-Check (qLDPC) results, one would only need to ensure that at most a constant fraction of the continuous modes are corrupted, and one will be able to have a constant fraction of logical continuous modes \cite{panteleev2022asymptotically,leverrier2022quantum,dinur2023good}.

Physically implementing quadratures which correspond to $x_1+27p_1$, for instance, could be quite challenging as this may correspond to a carefully calibrated balance of electric and magnetic fields or other dual charges. Fortunately, if one begins with a CSS code, there is always a CSS LDI representation that can be prescriptively generated \cite{gunderman2020local}. If one uses that method for generating an LDI representation one can be assured that the quadrature operators for the nullifiers are each only sums of $x$ or $p$ operators.


With this general result shown, we proceed now onto a non-trivial example. For this example, we will utilize some of the work from \cite{barnes2004stabilizer}, which in essence showed that if one had LDI generators one could make an analog CV code, but neither specified how to make LDI generators nor mentioned the distance of such codes in general, besides just stating their existence. However, the tools laid out for specifying the code space is of particular use here.

\begin{example}
We begin by recalling the following LDI form for the Steane code:
\begin{equation}
\phi_\infty(H_3)=
\setcounter{MaxMatrixCols}{15}
\begin{bmatrix}
 1 & 1 & 1 & 1 & 0 & 0 & 0 & | & 0 & 0 & 0 & 0 & 0 & 0 & 0\\
 0 & 1 & 1 & 0 & 1 & 1 & 0 & | & 0 & 0 & 0 & 0 & 0 & 0 & 0\\
 0 & 0 & 1 & 1 & 0 & 1 & 1 & | & 0 & 0 & 0 & 0 & 0 & 0 & 0\\
 0 & 0 & 0 & 0 & 0 & 0 & 0 & | & 1 & -1 & 1 & -1 & 0 & 0 & 0\\
 0 & 0 & 0 & 0 & 0 & 0 & 0 & | & 0 & 1 & -1 & 0 & 1 & -1 & 0\\
 0 & 0 & 0 & 0 & 0 & 0 & 0 & | & 0 & 0 & 1 & -1 & 0 & -1 & 1\\
\end{bmatrix}.
\end{equation}
As Pauli generators this code is given by:
\begin{equation}
    \langle XXXXIII,\ IXXIXXI,\ IIXXIXX,\ ZZ^{-1}ZZ^{-1}III,\ IZZ^{-1}IZZ^{-1}I,\ IIZZ^{-1}IZ^{-1}Z\rangle .
\end{equation}
As quadrature operators the nullifiers are then:
\begin{multline}
    \langle x_1+x_2+x_3+x_4,\ x_2+x_3+x_5+x_6,\ x_3+x_4+x_6+x_7,\\ p_1-p_2+p_3-p_4,\ p_2-p_3+p_5-p_6,\ p_3-p_4-p_6+p_7 \rangle
\end{multline}
We begin by finding the logical $|0\rangle$ state. We begin with the all $|0\rangle$ physical state and apply powers of the nullifiers to it. We then obtain:
\if{false}
This has a continuous parameter codespace corresponding to a single continuous mode in quadrature space. The mode (a single 2D subspace) is given by the states that satisfy:
\begin{eqnarray}
    e^{i(x_1+x_2+x_3+x_4)}|\psi\rangle&=&0\\
    e^{i(x_2+x_3+x_5+x_6)}|\psi\rangle&=&0\\
    e^{i(x_3+x_4+x_6+x_7)}|\psi\rangle&=&0\\
    e^{i(x_1-x_2+x_3-x_4)}\mathcal{F}[|\psi\rangle]&=&0\\
    e^{i(x_2-x_3+x_5-x_6)}\mathcal{F}[|\psi\rangle]&=&0\\
    e^{i(x_3-x_4-x_6+x_7)}\mathcal{F}[|\psi\rangle]&=&0,
\end{eqnarray}
where $\mathcal{F}[|\psi\rangle]$ is the Fourier transform of the position state $|\psi\rangle$. These $6$ equations fix $6$ of the $7$ position values, providing one mode of freedom.
\fi

\begin{eqnarray}
    |0_L\rangle&=&\int_{a,b,c} \Upsilon(a,b,c,d,e,f)|a,a+b,a+b+c,a+c,b,b+c,c\rangle\\
    &=& \int_{a,b,c,e} e^{i\cdot ce}|a,a+b,a+b+c,a+c,b,b+c,c\rangle.
\end{eqnarray}
In this instance the function $\Upsilon$ ended up being a function only of $c$ and $e$.

The logical operators are given by $\bar{X}=XXXXXXX$ and $\bar{Z}=ZZ^{-1}ZZ^{-1}ZZ^{-1}Z$. As quadratures these are:
\begin{equation}
    \bar{x}=x_1+x_2+x_3+x_4+x_5+x_6+x_7,\quad \bar{p}=p_1-p_2+p_3-p_4+p_5-p_6+p_7.
\end{equation}
Applying these with powers $\alpha$ for $\bar{x}$ and $\beta$ for $\bar{p}$, the logical mode values are then given by:
\begin{equation}
    |(\alpha,\beta)_L\rangle=\int_{a,b,c,e} e^{i\cdot ce}e^{i\beta\alpha}|a+\alpha,a+b+\alpha,a+b+c+\alpha,a+c+\alpha,b+\alpha,b+c+\alpha,c+\alpha\rangle,\quad \alpha,\beta\in \mathbb{R}.
\end{equation}
An error which is undetectable but moves within the codespace is given by: $p_5-2p_6+p_7$ or $x_5+x_6+x_7$. These each have weight $3$, as the results predict.
\if{false}
\begin{equation}
    \int \Upsilon(a,b,c,d,e,f)|a,a+b,a+b+c,a+c,b,b+c,c\rangle
\end{equation}
with:
\begin{eqnarray}
    \Upsilon&:=& e^{i[d(a-(a+b)+(a+b+c)-(a+c))+e(a+b-(a+b+c)+b-(b+c))+f(a+b+c-(a+c)-(b+c)+c)]}\\
    &=& e^{-2ice}
\end{eqnarray}\fi
\end{example}

Having covered this example, we now move on to the repeating cell case.




\subsection{The Continuous but Wrapping Case}

Here we consider a slight modification for analog CV codes where the phase space is either repeating cells in the phase space, or equivalently the quadratures are both upper-bounded by some value $p$--meaning that the values can be taken to be in the range $[0,p)$ \footnote{If they have unequal extent then the smaller of the two should be used for this result.}. Mathematically we are considering codes with bases which are continuous functions but which are supported in $[0,p)$. These will have codewords which are $L^2$-integrable functions with domain $[0,p)\times [0,p)$. It is worth noting that due to Fourier constraints \footnote{Note that the Fourier transform of a periodic, bounded function has support over the integers, and so both the function and its Fourier transform cannot be periodic. This result in generality is covered by the Pontryagin Theorem} \cite{stroppel2006locally}. this is not physically possible, however, we will consider this mathematical setting for now and consider the physical setting which corresponds best to this setting, \textit{i.e.} rotors, later in \ref{rotors}.

The following Corollary follows from applying the included proofs for Theorems \ref{ogproof} and \ref{cv}. In particular, the distance preservation argument from Theorem \ref{ogproof} still holds even if the underlying field is the reals modulo $p$, although artifact errors can now be induced if $p$ is not sufficiently large. Notably, the linear algebraic argument for avoiding the inducement of kernels still holds.


\begin{corollary}\label{boundedcase}
Let $\mathbb{R}_p^+$ be the non-negative real numbers below $p$, and let $p^*$ be some cutoff value. Then a stabilizer code with parameters $[[n,k,d]]_q$, for $q$ prime, can be transformed into a $[[n,k,d']]_{\mathbb{R}_p^+}$ code with $d'\geq d$ if $p>p^*$.
\end{corollary}

Given the current expressions for $p^*$, this would mean that the spatial extent needed would go as $B^{2(d-1)}(2(d-1))^{(d-1)}$ in order to guarantee the distance of the code, with $B$ being the maximal entry in the LDI representation, generally satisfying $B\leq \max_{i,j} |\phi_\infty(s_i)\odot\phi_\infty(s_j)|$. There is also the alternative expression given by $(B(q-1)(d-1)(1+(d-1)^2(q-1)^{d-1}(d-2)^{(d-2)/2}))^{d-1}$ \cite{gunderman2022degenerate} and in the case of CSS codes $B^{d-1}(d-1)^{(d-1)/2}$ \cite{moorthy2023local}. The values for $p^*$ can be quite large, however, depending on the level of precision achievable in the determination of the quadratures this may not be a particularly limiting factor. Settings for this may include quadratures for repeating cells or on physical rings (or loops). 

\begin{example}
In \cite{moorthy2023local} it was shown that $p^*=2$ for all members of the family of quantum codes derived from the classical Hamming family--those with parameters $[[2^N-1,2^N-1-2N,3]]_2$. Given that, for any value of $p\geq 2$ and $N\geq 3$ we may generate a family $[[2^N-1,2^N-1-2N,\geq 3]]_{\mathbb{R}_p^+}$, which could be a particularly promising early candidate if precision is limited.
\end{example}

As this is a spatial truncation of the general analog CV code case, we do not analyze any further examples here and instead proceed to the other specialized case of infinitely large bases but only discrete values.

\if{false}
Mathematically, we may interpret this a couple ways. Firstly, in terms of stabilizer formalism, we may consider this as operators:
\begin{equation}
    S\mapsto e^{i \pi S/p}
\end{equation}
Alternatively we may write this in terms of quadrature operators. In this case...

The proof is a given from the original result as linear dependency argument is unaltered.

Isn't this GKP in the case of $p=\sqrt{\pi}$ or so? Only if simple code? No, generally.
\fi

\subsection{The Formally Infinite, Discrete Case}

For this case we will actually consider two different settings. We will consider a discretized version of analog CV codes whereby the quadratures can only take discrete values. We will also consider a version more reminiscent of Pauli based stabilizer codes and discuss possible interpretations of this. Before we specialize to either formalism, we will prove general results for these codes. For notation, we will write both of these codes using the parameterized notation $[[n,k,d]]_{\mathbb{Z}}$ and note that the additive form mapping can be used in both directions to obtain analog-like codes as well as Pauli-like codes.

This is an example of stabilizer codes originating over a field and being still able to protect information over a mathematical ring. While multiplicative inverses may not exist for all members of the ring, there are still additive inverses, which still permits the correction of errors.

The following Theorem provides the proof for ensuring the distance of the code, while the generators are constructed through the same LDI framework:


\begin{theorem}\label{integercase}
Let $[[n,k,d]]_q$ be a stabilizer code, for a prime $q$, then we may transform it into a $[[n,k,d']]_\mathbb{Z}$ stabilizer code with $d'=d^*\geq d$.
\end{theorem}

\begin{proof}
A commuting set of generators is provided from the LDI representation. This proof centers around the same argument, whereby we only need to worry about artifact errors being induced. For any $(2t)\times (2t)$ matrix minor, corresponding to a Pauli of weight between $t$ and $2t$, the determinant will either be zero or not. The former case corresponds to either unavoidable errors or errors which already would have existed for the original code. The latter case has some nonzero integer value, which were we to consider the generators as being over the rationals, an extension ring of the integers, would have a matrix inverse and thus full rank. Given that it has full rank, there is no set of coefficients over the rationals which would reduce the rank, and therefore there are no coefficients over the integers which can reduce the rank, meaning that the rows and columns in this matrix minor are linearly independent. This then means that there is no new kernel induced, and so no artifact errors are induced. This then means that the distance of this code will be at least as good and by definition will be $d^*$.
\end{proof}

From this proof we trivially obtain the following Corollary which is less physically interpretable, but nonetheless mathematically defined:

\begin{corollary}
Let $[[n,k,d]]_q$ be a stabilizer code, for a prime $q$, then we may transform it into a $[[n,k,d']]_\mathbb{Q}$ stabilizer code with $d'=d^*\geq d$.
\end{corollary}

Of course, the codewords for a code over the integers is not normalizable, although taking a Gaussian envelope or something similar can generate highly accurate approximate codewords, much in the same fashion as GKP codewords \cite{gottesman2001encoding}. Taking codes over $\mathbb{Z}$ is very similar to a discrete, lattice-based version of the GKP encoding scheme. In this setting the $X$ and $Z$ operators act as translations in this discrete "phasespace". In this case, much like the GKP encoding, the codewords are points on a lattice although with nearest "taxicab"-distances ($\|\cdot\|_1$) between different lattice point codewords being $d^*$. This results in a phasespace separation between any pair of encoded logical states of at least $\sqrt{d^*}$, with additional distance provided by the powers on the logical operators, discussed further in Remark \ref{dps}.


For the Pauli-like generators in the infinite case, one should use Pauli operators as defined in \ref{infpaulis}, or other irrational induced phases from commuting $X$ and $Z$. This avoids trivializing the commutation relations, but still forms a (non-trivial) symplectic set of operators over discrete space.

For the quadrature interpretation the same additive form from Lemma \ref{addform} may be used. The primary change from the analog CV case is that all integrals are replaced by sums. In this setting the quadratures $x$ and $p$ act as discrete displacements by a unit value as opposed to continuous displacements. In this way these codes are strongly reminiscent of an ideal GKP code and may be a good setting for further studying GKP stabilizer-like encodings such as those in \cite{noh2020encoding}.

\begin{remark}\label{dps}
If we are interested in using these integer codes as codes for continuous variable systems, the distance metric provided for stabilizer codes is not the most prudent selection. The codewords for the integer code will form the lattice points in phase space assuming delta peaks, but the distance between the codewords is not simply described by the \textit{Pauli weight} of the logical operator typically since the powers of the operators are not assumed to be uniformly distributed. Instead, the distance in phase-space between the codewords will be given by the minimal geometric distance between the lattice of codewords. Let us write $d_{ps}$ as the distance in phase space, then:
\begin{equation}
    \sqrt{d}\leq d_{ps}:= \min_{\vec{a},\vec{b}\in \mathbb{Z}^n \ :\ X^{\vec{a}}Z^{\vec{b}}\in\mathcal{L}} \sqrt{\|\vec{a}\|^2+\|\vec{b}\|^2},
\end{equation}
with $\mathcal{L}$ being the logical operators for the integer code. An example is shown in Figure \ref{gkplattice}, assuming $\vec{b}=0$, this increases the lattice spacing by a factor of $\sqrt{14}\approx 3.74$. For the case of the Steane code the increased lattice spacing is $\sqrt{3}\approx 1.73$, as there is a weight-three CSS logical operator with all terms having power $1$. This scheme could aid in continuous variable encodings to discrete bases. This result is rather similar to those in \cite{conrad2022gottesman}, which provides a concrete and stabilizer-grounded approach to the same problem, providing the ability to use whichever stabilizer codes, although neglecting formal symplectic continuous variable analyses. Explicit unification could be valuable. Using good qLDPC codes in this provides an analytic solution for the open problem of a capacity achieving Gaussian encoding channel from \cite{harrington2001achievable}. In the work \cite{conrad2023good} they find a heuristic method for generating GKP-like encodings with phase space distance $\Theta(\sqrt{n})$ while having a constant rate. While there are also some interesting cryptographic applications from their work, it is a heuristic. This work is not a heuristic method but a deterministic, constructive method and so retains low weights for checks if derived from a "good" qubit qLDPC, especially if some weight reduction procedure such as \cite{hastings2021quantum} is applied before being transformed into this GKP-like encoding.

\begin{figure}[thb]
\includegraphics[width=\textwidth]{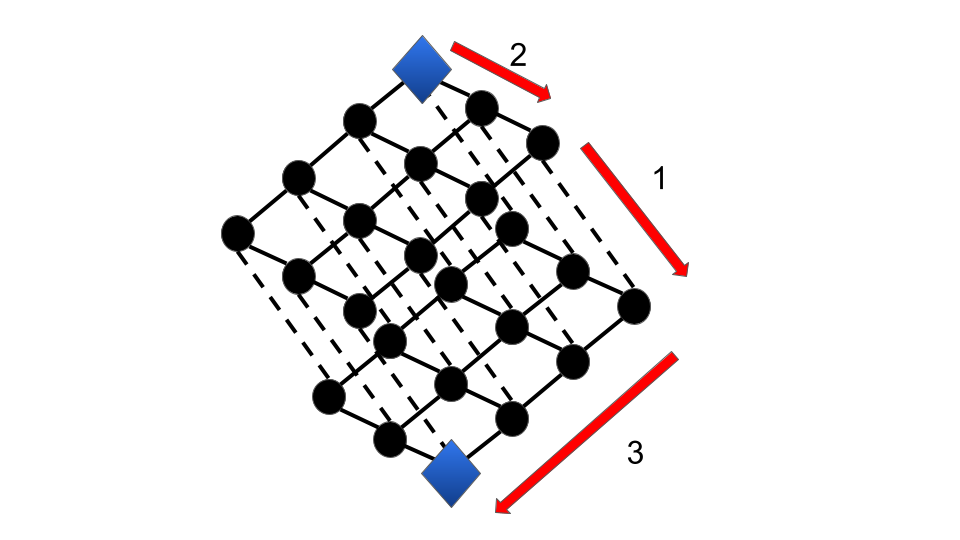}
\caption{Above is a schematic of an integer code, assuming a code with lattice points only differing by $X$ operators. The lattice would formally continue infinitely, but we only highlight this sublattice. Two codewords are placed on the blue diamonds. Their distance is $\sqrt{14}$, so their spacing is scaled by this factor and the inherent lattice spacing. This acts as an encoding of a discrete logical subspace in a possibly continuous phase space, or explicitly here a discretized phase space.}
\centering
\label{gkplattice}
\end{figure}
\end{remark}

We have now exhausted the various ways of restricting analog CV codes, aside from the trivial, and priorly studied, restriction to qudits. We finished these cases by considering codes formally over a mathematical ring, the integers, using a slight variation of the proofs used throughout this work. Next we take the idea further to show that the destination local-dimension can be a mathematical ring and the distance is preserved so long as the characteristic of the ring is sufficiently large.




\if{false}
\begin{example}
Use a qutrit code? Yes? No...
\begin{equation}
\setcounter{MaxMatrixCols}{27}
\begin{bmatrix}
 1 & 1 & 1 & 1 & 0 & 0 & 0 & | & 0 & 0 & 0 & 0 & 0 & 0 & 0\\
 0 & 1 & 1 & 0 & 1 & 1 & 0 & | & 0 & 0 & 0 & 0 & 0 & 0 & 0\\
 0 & 0 & 1 & 1 & 0 & 1 & 1 & | & 0 & 0 & 0 & 0 & 0 & 0 & 0\\
 0 & 0 & 0 & 0 & 0 & 0 & 0 & | & 1 & -1 & 1 & -1 & 0 & 0 & 0\\
 0 & 0 & 0 & 0 & 0 & 0 & 0 & | & 0 & 1 & -1 & 0 & 1 & -1 & 0\\
 0 & 0 & 0 & 0 & 0 & 0 & 0 & | & 0 & 0 & 1 & -1 & 0 & -1 & 1\\
\end{bmatrix}.
\end{equation}

\end{example}
\fi

\subsection{Ring Case}\label{ringthm}

Throughout this work we have considered a myriad of local-dimension cases, but now we reach the most exotic, and general, case in this work. We will show that any stabilizer code originating over a field local-dimension can be used as an additive quantum code, much like a stabilizer code, over mathematical rings and under specific conditions have the distance of this code be at least as large. As before, we will prove our general result then show an example to highlight a possible construction. Some example uses of this result for codes include \cite{ashikhmin2001nonbinary,niehage2005quantum,ketkar2006nonbinary,grimsmo2020quantum,albert2020robust,faist2020continuous}, the majority of which can be obtained from, or supplemented by, the results herein. Beyond these, many anyonic condensation rules can be utilized within this setting, providing a framework for combining them with already known stabilizer codes. 

For this result it helps to briefly discuss fields of fractions for commutative rings with integral domains. While the integers are a mathematical ring, the associated field of fractions for the integers is the field of rational numbers. This concept holds more broadly for certain mathematical rings. For an integral domain $\mathcal{R}$, which is a commutative ring within which $ab=ac$ ($a\neq 0$) implies $b=c$, there is an associated field of fractions to it, denoted $Q(\mathcal{R})$, formed from ratios of elements of $\mathcal{R}$ and $\mathcal{R}-\{0\}$. The field of fractions is a field.

\begin{theorem}[$\mathcal{F}\mapsto \mathcal{R}$ Case]
Let $\mathcal{F}$ be a finite field and $\mathcal{R}$ be an integral domain. Then if $\text{char}(\mathcal{R})=0$ or $\text{char}(\mathcal{R})>p^*$, with $p^*$ some cutoff value that is a function of $\text{char}(\mathcal{F})$, $n$, $k$, and $d$, a stabilizer code $[[n,k,d]]_{\mathcal{F}}$, with $r$ linearly independent generators in the $\phi$ representation, can be used as a $[[n,n-r,d']]_{\mathcal{R}}$ stabilizer code with $d'\geq d$.
\end{theorem}

Before proving this Theorem it is worth noting that all finite integral domains are isomorphic to finite fields \footnote{This can be seen from the requirement that for $c\neq 0$, $ac=bc$ implies $a=b$. Then due to being finite for each $r\in \mathcal{R}$ there must be some $n>m$ such that $r^n=r^m$ or equivalently $r^{n-m}=1$, therefore each nonzero element has a multiplicative inverse and thus $\mathcal{R}$ is a field.}.

\begin{proof}
Theorem \ref{fldi} provides a set of generators preserving full rank even upon the matrix rows being considered to generate a module. For the distance, one separates the set of undetectable errors into unavoidable errors and artifact errors and avoids inducing any artifact errors which would reduce the distance of the code. Consider errors with support in the symplectic module formed from $2j\times 2j$ entries, corresponding to support overlap with Pauli operators with weight $j$, randomly permuted. This module considered as a matrix may have zero determinant, which corresponds to either unavoidable errors or errors which already existed over $\mathcal{F}$. The zero value case is handled as before. Considering first the case of finite integral domains, where we use the isomorphism to finite fields. If the determinant is some value below the characteristic of the ring $\mathcal{R}$ then it will be some nonzero element in $\mathcal{R}$. If one considers the associated field of fractions for the ring, $Q(\mathcal{R})$, the matrix will be invertible, meaning that no new kernel is introduced, so the distance cannot be decreased in this way. Lastly a kernel can be induced if the determinant is a multiple of the characteristic of the ring $\mathcal{R}$. We may avoid this possibility by ensuring the characteristic of $\mathcal{R}$, $\text{char}(\mathcal{R})$, satisfies:
\begin{equation}
    \text{char}(\mathcal{R})>p^*=B^{2(d-1)}(2(d-1))^{(d-1)},
\end{equation}
where the expression is obtained from Hadamard's bound on the determinant of a $2(d-1)\times 2(d-1)$ matrix with maximal entry $B$. Therefore the distance in this case of the code will be at least as large as the original code's distance.

For infinite integral domains, all nonzero elements that the determinant may evaluate to has an associated inverse element in the field of fractions, $Q(\mathcal{R})$, and thus the rank of the module cannot be lowered. Therefore no artifact errors can be introduced and so the distance will be at least as large.
\end{proof}

In the above, one aspect worthy of some extra discussion is the expression for $p^*$, particularly the maximal entry $B$. If the field has finite characteristic then the prior bounds for $B$ hold. If, however, the characteristic for the field is infinite, then one may replace the terms in the expression for $B$ that are $(q-1)$ with the maximal value used. So for instance, were one using a code which originated over the rationals, $\mathbb{Q}$, and the largest entry in the canonical representation is $777/17$, then the $(q-1)$ terms would be replaced by this fraction in the bound. By adding this caveat, one may begin over a field local-dimension value, which does not have a finite characteristic. Alternatively, one may use the alternative bound for $B$ based on the maximal symplectic product, if that LDI representation is used \cite{gunderman2022degenerate}.


Some examples of other integral domains include: fields $\mathcal{F}$, prime-power local-dimensions $p^m$, polynomials, power series, and polynomials quotiented by prime ideals, $p$-adic integers, and the rings $\mathbb{Z}[\sqrt{n}]$ for non-square integers $n$. Note that the errors must also be from the same mathematical object as the code is written. While most of these local-dimensions are not physically observed, were one to find such a system or a theoretical framework where these would be considered, this would permit the distance promises to apply immediately.

\begin{example}
In \cite{moorthy2023local} it was shown that $p^*=2$ for all members of the family of quantum codes derived from the classical Hamming family--those with parameters $[[2^N-1,2^N-1-2N,3]]_2$. Given that, for any integral domain $\mathcal{R}$ with $\text{char}(\mathcal{R})\geq 2$, or $\text{char}(\mathcal{R})= 0$, and $N\geq 3$ we may generate a family $[[2^N-1,2^N-1-2N,\geq 3]]_{\mathcal{R}}$, so it could be a particularly promising early candidate for systems with this variety of local-dimension.
\end{example}

\if{false}
\begin{example}
For this example, we will again consider the LDI Steane code. Let us select as the local-dimension the quadratic integer ring $\mathbb{Z}[\sqrt{-7}]$, which are the values $\{a+\sqrt{7}bi,\ a,b\in \mathbb{Z}\}$. Still trivial...

Still debating about local-dimension choice which will be interesting and has not been discussed/done before...
\end{example}

\begin{example}
Do an example with the ring of circulant matrices. It should work. As long as the matrices are sufficiently large. Note that the errors here must also be considered as circulant matrices! So blocks of Pauli errors.
\end{example}
\fi



\section{Future Directions}\label{concl}

In the cases where the distance of the code is only promised for sufficiently large local-dimensions, this method may be somewhat limited in its use. Below this value one would need to determine the distance through other means. One particularly readily analyzable setting is that of topological codes. We will briefly delve into this case now.

\subsection{The Toric Code and Similar Topological Codes}

All of the examples considered so far have been traditional stabilizer codes, however, topological codes can also be written in matrix form. In more near-term applications topological quantum error-correcting codes have become of greater interest. Additionally, topological codes require some sense of spatially neighboring qubits, which we show can be leveraged here. Amongst the many topological codes is the toric code, proposed by Kitaev, which has a distance that goes as the square root of the number of qubits \cite{kitaev2003fault}. We examine the toric code here and show how we can put it into an LDI form that always preserves the distance, even if the local-dimension is a general commutative ring. This code was already considered in \cite{gunderman2022collective}, but we put this result into the context of this work.

The toric code is constructed on a square grid of qubits and its dual square grid also as qubits. The code is CSS with the stabilizer generators given as $X$ operators on the edges of each grid square, and $Z$ operators on the cross from the intersection of squares, see Figure \ref{toricstabs}. We make the proposed modifications in the figure, flipping the power of opposing pairs of Pauli operators within each generator. This allows for compositions to still perfectly commute with all other generators. This simple change provides an LDI form for the code, but does not, yet, guarantee that the distance is preserved.

For the toric code, the logical operators are completely topologically winding strings of a single Pauli operator variety. This property gives the code the high distance as the weight of the logical operators is a lower bound on the distance for non-degenerate codes such as this. As an ansatz, we take these same logical operators and consider the coset equivalent forms for them upon interacting with the stabilizer generators. In Figure \ref{toriclogicals} we show two characteristic examples. The top example imagines applying the $X$ stabilizer from below the middle $X$ operator, which deforms the string. This is exactly the same action as the qubit case, except the left edge is a $-1$ power instead of $+1$. Consider, however, if we apply the same generator from above. Then the string is deformed like the qubit case but there is also an added \textit{bridge} in the operator that is impossible in the qubit case (it disappears upon taking $\mod 2$). Now, this behavior could be extended and we could generate more bridges, however, no bridge can fully wind around the torus and generate a non-trivial action on the logical space. This means that the distance is still preserved, although these peculiar bridges appear for the non-qubit case. As of this time, it is not clear whether these bridges provide some possible other changes to the toric code, aside from singular strings not always being the logical operator, but rather strings with possible multiplicities in them.

\begin{figure}[!htb]
\begin{center}

\includegraphics[width=\textwidth]{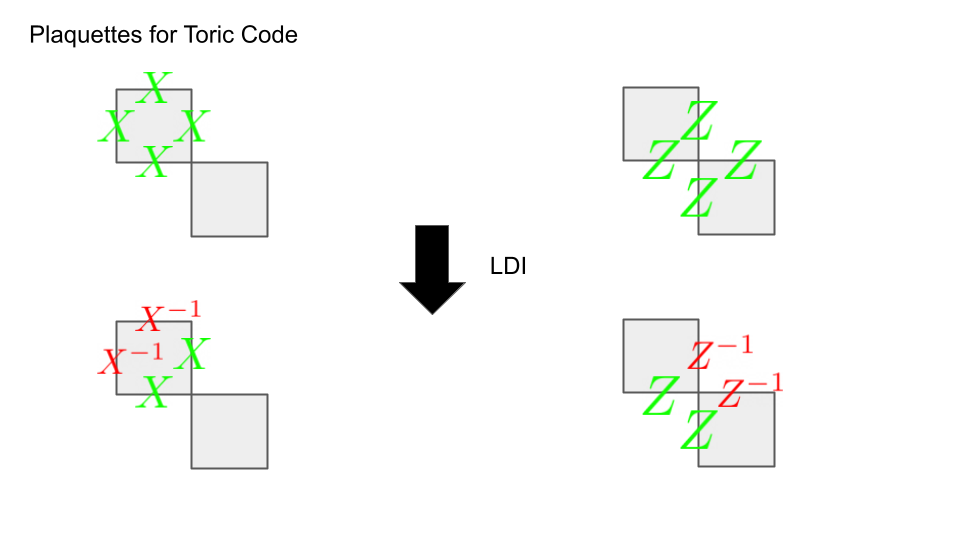}
\end{center}
\caption{This figure shows the transformations on the stabilizing plaquettes for the qubit toric code to generate a toric code that works regardless of the local-dimension of the system. This transforms the code from a $[[2N^2,2,N]]_2$ code into a $[[2N^2,2,N]]$ code for any choice of local-dimension with additive inverses.}
\centering
\label{toricstabs}
\end{figure}

\begin{figure}[!htb]
\begin{center}
\includegraphics[width=\textwidth]{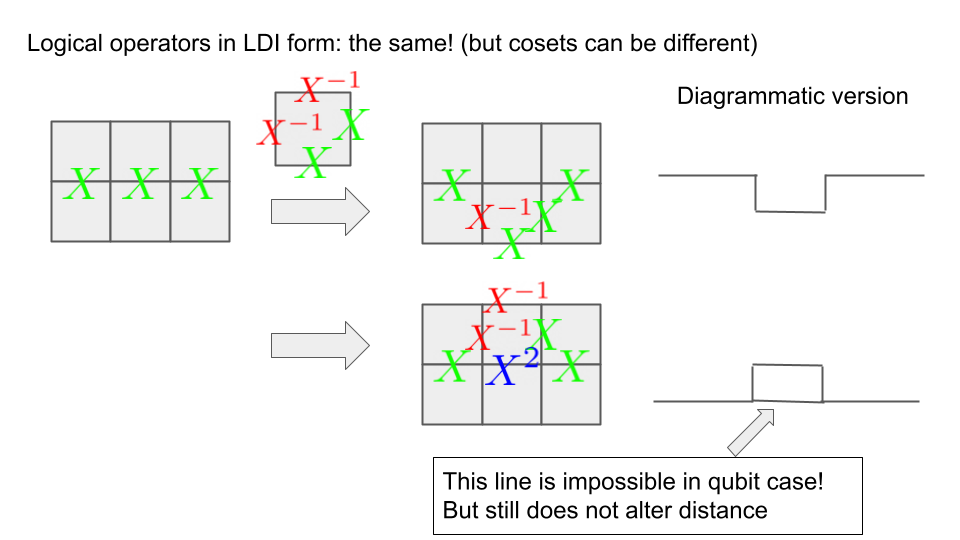}
\end{center}
\caption{This figure illustrates the distance argument, showing that the distance of the code is preserved upon changing the local-dimension of the code. Notice that bridges generated are only existent in the non-qubit case, as modulo $2$ they disappear and a mere distortion of the logical string operator occurs}.
\centering
\label{toriclogicals}
\end{figure}

A couple of notable features occur in this example which we can take for more general topological codes. Firstly, the generators and the logical operations required no properties of the underlying mathematical object other than additive inverses, so arbitrary rings can be used for the local-dimension, including composite valued local-dimensions and groups. This then provides the codes from \cite{watanabe2023ground,mathur2022n} as well as codes for broader mathematical objects than those considered. Having spatially local LDI generators ensures that one only has to satisfy the distance promise for subsets of these local-operators which provides for the global distance being preserved. Unfortunately as of this time there are no results on whether the same generators, just with uniform power changes for each generator, can always be obtained for topological codes, nor are there bounds on the level of locality loss from using a prescriptive method. However, this may prove a fruitful direction for future work. Additionally, the toric code example considered satisfies fusion rules for certain anyons. It is likely that changes to the powers in the generators and the shapes of the generators could obtain the results of \cite{ellison2022pauli}. This could aid in extending these results to non-topological codes and will also put these results more clearly within the context of stabilizer codes.

\subsection{Rotor Codes}\label{rotors}

An interesting possible use case for this work can be found in rotor systems and their codes. A rotor system has its two Pauli axes having different varieties of local-dimensions, in particular one is the truncated continuous space $\mathbb{R}_{2\pi}^{+}:=[0,2\pi)$, while the other is the integers $\mathbb{Z}$. There are a few physical settings where rotors best describe the underlying physical system. We direct interested readers to \cite{vuillot2023homological} for details on this. That work also described methods for generating codes for these systems through homological products, with some cases requiring the LDI condition. This work permits direct algebraic importing of stabilizer codes. In this work, we may combine Corollary \ref{boundedcase} and Theorem \ref{integercase} to cover systems like this, however, the distance for both axes can only be promised if $p^*<2\pi$. This then leads to the following conjecture:
\begin{conjecture}
Let $\phi_\infty(\mathcal{S})$ represent a good qubit LDPC code (with constant rate and relative distance). If $p^*<2\pi$, $\phi_\infty(\mathcal{S})$ represents a good quantum LDPC rotor code.
\end{conjecture}
If the basis described by $\mathbb{Z}$ is more appropriately described by some finite portion of $\mathbb{Z}$ then the appropriate Theorem and $p^*$ value should be used. The strongest way to show this conjecture would be to show that for a code with parameters $[[n,k,d]]_q$, the LDI form has its distance at least preserved for $p^*=q$, which would permit codes with local-dimensions $2$, $3$, and $5$ to be used freely for rotor systems. However, showing this only for $q=2$ would suffice for the above conjecture. The sparsity will be retained as the number of nonzero entries in $L$ will be sparse. While the traditional stabilizer distance $d$ is not always the best metric for continuous spaces, this value still provides information on what fraction of the registers may be lost and the information still reliably recovered.

\subsection{Conclusion}

In this work we have shown that the LDI extension of the traditional stabilizer formalism can be used to obtain some different settings of stabilizer-like codes. We showed that they permit the use of arbitrary stabilizer codes as analog CV codes and that these codes will keep their parameters. Additionally we considered the restriction to bounded phase space and integer precision, each of which correspond to different approximations of analog CV codes. In the case of only integer precision, forming a lattice, we find that the codes are remarkably reminiscent of GKP codes. Using these ideas we showed that stabilizer codes originating over a field local-dimension can have their distance still kept when used over integral domain local-dimension systems. These are of theoretical interest for quantum error-correcting code theorists, as they may exhibit interesting properties not observable for field local-dimension codes. These mathematical ring local-dimension codes can also be of use for physical systems whose Pauli operators best correspond to these mathematical objects.

Naturally, throughout this work one can pull in qLDPC results, although the sparsity, $\rho$ nonzero entries at most per register or generator, would be roughly squared. These good qLDPC codes would still be good qLDPC codes upon being used in an LDI form since they would still have a constant number of nonzero entries as the index of the family increases. One may expect that a weight reduction technique similar to that of Hastings would exist in these other local-dimensions, and so could remove this quadratic increase in the check weights \cite{hastings2021quantum}.

The full implications of the results in this work have not been explored and likely contain many other useful and interesting cases. Some possible future directions to carry this work would include connecting the LDI framework to anyonic theories, subsystem codes, and GKP stabilizer codes, among other quantum error-correcting frameworks. It is possible that some of the other methods for quantum error-correction can be unified through the LDI framework. For more near-term uses properties such as fault-tolerance preservation could be useful features to show for LDI representations of codes. Additionally, for those results where $p^*$ was involved, possible reductions in the expressions for this could prove very helpful. Lastly, we were able to show the ability to apply stabilizer codes over integral domains, but it could be interesting to extend the results to general commutative rings, which would cover composite local-dimension values as well as some other exotic possible local-dimensions.

\section*{Acknowledgments}

We thank the following people for helpful discussions and ideas: E. Knill, A. Jena, M. Vasmer, G. Dauphinais, V. Albert, B. Terhal, C. Vuillot, and A. Ciani.





\if{false}

\clearpage

\begin{lemma}
Let $N$ be a positive integer with prime factorization $\prod_{i=1}^j p_i^{a_i}$ with $a_i>0$ and $p_1< p_2<\ldots p_j$. If $p_1>p^*$ then $d'\geq d$.
\end{lemma}

\begin{proof}
This follows from the original proofs, since this selection of local-dimension cannot induce new artifact errors.
\end{proof}

Notable, if $N$ satisfies this requirement, $N+1$ would require $p^*=2$ in order to preserve its distance since $2$ will then be a factor.

\begin{lemma}
Given a cyclic \textbf{qubit} quantum code $\mathcal{C}$, there exists a local-dimension-invariant representation for $\mathcal{C}$ which is still cyclic.
\end{lemma}

\begin{proof}
Begin with the code over local-dimension $2$. We may define row addition in a cyclic form: . Then we may add row $1$ to all rows after it. 
\end{proof}

\begin{proof}[General cyclic case proof]
Let the rows be formed by cyclic shifts, $\delta$, of $c$ so that all rows are given by $\delta^i c$. Then in the stabilizer code space then is a composition of generators such that $(\prod_k \delta^k c)c$ forms the first row of a canonical representation for the code. Since each generator is a shift of $c$, there will be corresponding independent generators such that the cyclic nature is preserved and the code is in canonical form. The code's structure, in the $\phi$ representation will appear as:
\begin{equation}
    [I_{n-k}\ X(c)\ |\ Z_1(c)\ Z_2(c)],
\end{equation}
where the $c$ arguments are there to indicate the cyclic entries. Finally we may add the cyclic version of the symplectic matrix over the integers:
\begin{equation}
    L_{i,(i+j)}=\phi_\infty(s_i)\odot \phi_\infty(s_{i+j\mod (n-k)}),\quad j\leq n-k-1.
\end{equation}
One could alter some strips (values of $j$) of $L_{ij}$ so that only positive or only negative terms are used throughout the entire $L$ matrix.
\end{proof}

\begin{lemma}
Given a cyclic qudit quantum code $\mathcal{C}$, there exists a local-dimension-invariant representation for $\mathcal{C}$ which is quasi-cyclic with index 2.
\end{lemma}

\begin{proof}
This uses alternating cyclic addition (one to add, one to subtract).
\end{proof}

\begin{proof}
For this proof we will use the following cyclic preserving way of adding rows...
\end{proof}

\begin{theorem}
Take $[[n,k,d]]_q$ to $[[n,k,d]]_\infty$ (doubly infinite case, then special case of singly infinite (qho)).[Is this a $C*$? probably not since not infinite registers]
\end{theorem}

\begin{proof}
For this we define the following operators:
\begin{equation}
    X|j\rangle =|j+1\rangle,\quad Z|j\rangle=\omega_\infty^j|j\rangle,\quad \omega_\infty=e^{2\pi i/\sqrt{2}},
\end{equation}
[maybe $\mathbb{Z}$ instead of $\infty$?] note that we selected $\sqrt{2}$ as the denominator, but any irrational number would suffice. This is since we wish for $\omega_\infty$ to have unbounded multiplicative order, but also inverses. $p^*$ is achieved automatically.
\end{proof}

\begin{theorem}
Singly infinite case. Semigroup only. Is this possible? What do the errors look like?
\end{theorem}

\begin{proof}
Assume that there is nearly no chance that $(a_i^\dag- a_i)<m$ then use bases $[m,\infty)$. Probably should consider model as errors $a_i^d$ and... not sure yet.
\end{proof}

\begin{theorem}
LD ignorant computation--sum of spaces, like TC
\end{theorem}

spatial locality?

transversality?

\begin{definition}
Bosonic error model: powers of $a$ count as higher weights, so that $pr(a^d)\approx 0$. This is mostly an artifact of not being doubly infinite. Likewise do the same for $Z$ type errors, picking $\omega$ so that $arg(\omega^d)<2\pi $ (or $\omega^{d^2}$? so that each of $d$ registers can be damaged up to $d$ times?) with an irrational frequency. High $d$ just means needing higher resolution.
\end{definition}

\begin{theorem}
Working in $\mathbb{C}^n$, in coherent states, packing a lot of states around an annulus. Distance is maybe $dr$ with code distance $d$ and circle radius $r$. Then construct them together. Protects again $D(a)$ and $D(bi)$ displacements and $a$ (photon loss). Must construct an explicit example.
\end{theorem}

Note that $2r>R\cdot 2\pi/q>r$ for $q$ circles of (maximally chosen) radius $r$ with $R$ being the distance from the origin to the center of each circle. The lower bound is due to the annulus' center needing to traverse from the center of the circle to its boundary, while the upper bound is from this arc not fully spanning a diameter. The point to point distance is what we consider here.

As our example, let's consider a Shor-like code, utilizing a LDI representation for this code:
\begin{equation}
    \langle XX^{-1}II^6,\ IXX^{-1}I^6,\ I^3XX^{-1}II^3,\ I^3IXX^{-1}I^3,\ I^6XX^{-1}I,\ I^6IXX^{-1},\ Z^6I^3,\ I^3Z^6\rangle
\end{equation}

What are $X,Z$ in coherent space? $X$ is rotation by $2pi/q$ and $Z$ is a phase shift between coherent states. But what about radial corrections? Renormalization?

How to do syndrome measurements? It would indicate just which other circle in the annulus we're in right? In effect it indicates which symmetry change or rotation has occurred?


Note that the generators of an LDI code commute perfectly. The distance promise only requires that the LD is bigger than $p^*$--so this includes $\mathbb{R}$, $\mathbb{C}$, Galois rings, and any other ring such that the largest element is greater than $p^*$. More generally:

\begin{theorem}
Let $\mathcal{S}$ be a QECC over some field $\mathbb{F}$, or a code that is already expressed in canonical form over some ring $\mathbf{R}$. Then we may generate a QECC with parameters $[[n,k,d']]_{\mathbf{R}'}$ for any $\mathbf{R}'$ with modulus greater than $p^*$.
\end{theorem}

modulus for Galois ring: $GR[p,f(x)]$ need $f(p)>p^*$.

non-comm is ok. but need map from Z to R. like lifted-product where ints to circulants/comm ring. 

Sparsity preservation level


Cite Vlad's paper that kinda does it, but then we added in LDI. Would this also hold for rings too then? Using Jordan normal forms? Give examples using circulant matrices and some noncomm ring of matrices

Using LCM can make this true for composites too as initial local-dimension. $lcm(m_1,m_2,\ldots)\mod D\neq 0$ case is given

\begin{lemma}
Let $\mathcal{S}_c$ be a stabilizer code with initial local-dimension $c$, where $c$ is composite number, and with a canonical form $[M\ X_2\ |\ Z_1\ Z_2]$, where $M$ is a diagonal matrix with entries $m_1\leq\ldots  \leq m_{n-k}$, such that $lcm(m_1,\ldots , m_{n-k})\mod D\neq 0$ (this is a given actually, so provides all cases). Then $\mathcal{S}_c$ can be put into an LDI representation. 
\end{lemma}

This canonical form does not preserve the distance... so I must show this myself...

Geez gets really rough really quickly... modules and such come up...

Note that this does not rule out the possibility that other cases of composite codes can be put into an LDI representation, and in a prescriptive method, but merely provides one case whereby they can be provided.

\begin{proof}
Let $\mathcal{S}_c=[M\ X_2\ |\ Z_1\ Z_2]$. Now construct the following LDI representation: 
\end{proof}

\begin{example}
\begin{equation}
    \mathcal{S}_6=\begin{bmatrix}
    4 & 0 & 1 & | & 1 & 1 & 3\\
    0 & 5 & 1 & | & 0 & 1 & 3
    \end{bmatrix}
\end{equation}
\end{example}

\fi

\if{false}

Having protected quantum information is an essential piece of being able to perform quantum computations. There are a variety of methods to help protect quantum information such as those discussed in \cite{lidar2013quantum}. In this work we focus on stabilizer codes as they are the quantum analog of classical linear codes. Even with error-correcting codes, having sufficient amounts of protected quantum information to perform useful tasks is still an unresolved challenge. A way to retain a similarly sized computational space while reducing the number of particles that need precise controls is to replace the standard choice of qubits with \textit{qudits}, quantum particles with $q$ levels, also known as local-dimension $q$ \cite{wang2020qudits}. Throughout this work we require $q$ to be a prime so that each nonzero element has a unique multiplicative inverse over $\mathbb{Z}_q$. This restriction can likely be removed, but for simplicity and clarity we only consider this case. Experimental realizations of qudit systems are currently underway \cite{low2020practical,quditlight,kononenko2021characterization,yurtalan2020implementation}, so having more error-correcting codes will aid in protecting such systems.

Prior work on qudit error-correcting codes have at times had challenging restrictions between the parameters of the code \cite{quditgeneral,quditbch,quditmds}, and we've already made progress on reducing this barrier in a prior paper \cite{gunderman2020local}. Our prior work showed the ability to make error-correcting codes that preserved their parameters even upon changing the local-dimension of the system, provided the local-dimension is sufficiently large. Unfortunately the ability to promise the distance of the codes was only shown for non-degenerate codes and with a large local-dimension value required. Beyond this, qudits also have proven connections to foundational aspects of physics \cite{contextuality}. Seeing these potential reasons for using qudits, this work builds off of our prior work to expand the local-dimension-invariant framework to the case of degenerate codes, as well as providing a roughly quadratic improvement in the size of the local-dimension needed to still promise the distance of the code. With these results the practicality of using this method is improved as well as now providing the option of applying the result to the essential class of degenerate codes, such as quantum versions of low-density parity-check (LDPC) codes.


In this section we review some key facts about qudit stabilizer codes. For a more complete guide on qudit stabilizer codes, we recommend \cite{quditgeneral}. The definitions laid out here will be used throughout this work. Let $q$ be the local-dimension of a system, where $q$ is a prime number. We will denote by $\mathbb{Z}_q$ the set $\{0,1,\ldots q-1\}$. When $q=2$ we refer to each register as a qubit, while for any value of $q$ we call each register a qudit. In order to speak more generally and not specify $q$, we will often times refer to each register as a particle instead. We now begin to define the operations for these registers.

\begin{definition}\label{dim}
Generalized Paulis for a particle over $q$ orthogonal levels (local-dimension $q$) are given by:
\begin{equation}
 X_q|j\rangle=|(j+1)\mod q\rangle,\quad Z_q|j\rangle=\omega^j|j\rangle
\end{equation}
with $\omega=e^{2\pi i/q}$, where $j\in\mathbb{Z}_q$. These Paulis form a group, denoted $\mathbb{P}_q$.
\end{definition}

When $q=2$, these are the standard qubit operators $X$ and $Z$, with $Y=iXZ$. This group structure is preserved over tensor products since each of these Paulis has order $q$. A generalized Pauli over $n$ registers is a tensor product of $n$ generalized Pauli group members over a single register. 

A commuting subgroup of generalized Pauli operators with $n-k$ generators, but not including any nontrivial coefficient for the identity operator, is equivalent to a stabilizer code. The number of orthogonal eigenvectors, which form bases called codewords, for these $n-k$ generators is $q^k$. In effect, we have constructed $k$ \textit{logical} particles from the $n$ \textit{physical} particles. If we are to use these subgroups for error-correction purposes then they ought to be able to have some accidental operator occur and still have the codewords be discernible.
We will work under the assumption that errors on distinct particles are independent and we will assume the error model on each qudit is the depolarizing channel.\if{false}, which we define as follows:
\begin{equation}
    \mathcal{E}(\rho) = (1-p)\rho + \frac{p}{q^2-1} \sum_{E \in \mathbb{P}_q \setminus \{I\}} E \rho E^\dagger
\end{equation}\fi
Given this error model we will predominantly be interested in the number of non-identity terms in any error as the exponent of the error term increases with this.

\begin{definition}
The weight of an $n$-qudit Pauli operator is the number of non-identity operators in it.
\end{definition}

\begin{definition}
A stabilizer code, specified by its $n-k$ generators, is characterized by the following set of parameters:
\begin{itemize}
\item $n$: the number of (physical) particles that are used to protect the information.
\item $k$: the number of encoded (logical) particles.
\item $d$: the distance of the code, given by the lowest weight of an undetectable generalized Pauli error. An undetectable generalized Pauli error is an $n$-qudit Pauli operator which commutes with all elements of the stabilizer group, but is not in the group itself.
\end{itemize}
These values are specified for a particular code as $[[n,k,d]]_q$, where $q$ is the local-dimension of the qudits.
\end{definition}

We pause for a moment here to discuss how degenerate codes differ from non-degenerate codes. Degenerate codes are different in the following equivalent ways. Firstly, they may have multiple errors with the same syndrome value and that map to different physical states, but upon recovery still map back to the same logical state. Secondly, degenerate codes may have generators, aside from the identity operator, which have lower weight than the distance of the code. These two differences make degenerate codes markedly different from their non-degenerate counterpart. Degenerate codes, while having these extra nuances, are a crucial class of stabilizer codes as any quantum analog of a low-density parity-check (LDPC) code with high distance will need to be a degenerate code. We will begin our new results by focusing on non-degenerate codes, then move to the degenerate case in Theorem \ref{degen}, however, there are more tools needed before discussing the new results.


Working with tensors of operators can be challenging, and so we make use of the following well-known mapping from these to vectors, following the notation from \cite{gunderman2020local}. This representation is often times called the symplectic representation for the operators, but we use this notation instead to allow for greater flexibility, particularly in specifying the local-dimension of the mapping. This linear algebraic representation will be used for our proofs.

\begin{definition}[$\phi$ representation of a qudit operator]
We define the linear surjective map: 
\begin{equation}
\phi_q: \mathbb{P}_q^n\mapsto \mathbb{Z}_q^{2n}
\end{equation}
which carries an $n$-qudit Pauli in $\mathbb{P}_q^n$ to a $2n$ vector mod $q$, where we define this mapping by:
\begin{equation}
I^{\otimes i-1} X_q^a Z_q^b I^{\otimes n-i} \mapsto \left( 0^{i-1}\ a\ 0^{n-i} \middle\vert 0^{i-1}\ b\ 0^{n-i}\right),
\end{equation}
which puts the power of the $i$-th $X$ operator in the $i$-th position and the power of the $i$-th $Z$ operator in the $(n+i)$-th position of the output vector. This mapping is defined as a homomorphism with: $\phi_q(s_1\circ s_2)=\phi_q(s_1)\oplus \phi_q(s_2)$, where $\oplus$ is component-wise addition mod $q$. We denote the first half of the vector as $\phi_{q,x}$ and the second half as $\phi_{q,z}$.
\end{definition}

We may invert the map to return to the original $n$-qudit Pauli operator with the global phase being undetermined. We make note of a special case of the $\phi$ representation:

\begin{definition}
Let $q$ be the dimension of the initial system. Then we denote by $\phi_\infty$ the mapping:
\begin{equation}
    \phi_\infty:  \mathbb{P}_q^n\mapsto \mathbb{Z}^{2n}
\end{equation}
where no longer are any operations taken $\mod$ some base, but instead carried over the full set of integers.
\end{definition}

The ability to define $\phi_\infty$ as a homomorphism still (and with the same rule) is a portion of the results of \cite{gunderman2020local}. $\phi_q$ is the standard choice for working over $q$ bases, however, our $\phi_\infty$ allows us to avoid being dependent on the local-dimension of our system when working with our code. Formally we will write a code in $\phi_q$, perform some operations, then write it in $\phi_\infty$, then select a new local-dimension $q'$ and use $\phi_{q'}$. We shorten this to write it as $\phi_\infty$, and can later select to write it as $\phi_{q'}$ for some prime $q'$ by taking element-wise $\mod q'$. While the operators in $\phi_\infty$ all commute, normalization of the codewords for infinitely many levels becomes a potential problem.

The commutator of two operators in this picture is given by the following definition:
\begin{definition}
Let $s_i,s_j$ be two qudit Pauli operators over $q$ bases, then these commute if and only if:
\begin{equation}
\phi_q(s_i)\odot \phi_q(s_j)=0\mod q
\end{equation}
where $\odot$ is the symplectic product, defined by:
\begin{multline}
\phi_q(s_i)\odot \phi_q(s_j)\\ =\oplus_k [\phi_{q,z}(s_j)_k\cdot  \phi_{q,x}(s_i)_k- \phi_{q,x}(s_j)_k \cdot \phi_{q,z}(s_i)_k]
\end{multline}
where $\cdot$ is standard integer multiplication $\mod q$ and $\oplus$ is addition $\mod q$.
\end{definition}

When the commutator of $s_i$ and $s_j$ is not zero, this provides the difference in the number of $X$ operators in $s_i$ that must pass a $Z$ operator in $s_j$ and the number of $Z$ operators in $s_i$ that must pass an $X$ operator in $s_j$ when attempting to switch the order of these two operators.

Before finishing, we make a brief list of some possible operations we can perform on our $\phi$ representation:
\begin{enumerate}
    \item We may perform elementary row operations over $\mathbb{Z}_q$, corresponding to relabelling and composing generators together.
    \item We may swap registers (qudits) in the following ways:
        \begin{enumerate}
            \item We may swap columns $(i,i+n)$ and $(j,j+n)$ for $1\leq i,j\leq n$, corresponding to relabelling qudits.
            \item We may swap columns $i$ and $(-1)\cdot (i+n)$, for $1\leq i\leq n$, corresponding to conjugating by a Hadamard gate on particle $i$ (or Discrete Fourier Transforms in the qudit case \cite{qudit}) thus swapping $X$ and $Z$'s roles on that qudit.
        \end{enumerate}
\end{enumerate}

All of these operations leave the code parameters $n$, $k$, and $d$ alone, but can be used in proofs.


In this section we recall the results relating to local-dimension-invariant (LDI) stabilizer codes. These codes answer the question of when we can apply a code from one local-dimension $q$ on a system with a different local-dimension $p$. While an unusual property, an LDI code would permit the importing of smaller local-dimension codes for larger local-dimension systems. Some codes with particular parameters may not be known and so this fills in some of these gaps. Additionally, this framework could potentially provide insights into local-dimension-invariant measurements. Few examples of LDI codes, although not by this name, were known, notable the 5 particle code \cite{5qudit} and the 9 particle code \cite{chau1997correcting}, until the recent work in \cite{gunderman2020local} which showed that all codes can satisfy the commutation requirements, and at least for sufficiently large local-dimensions the distance can also be at least preserved. We will review next the primary results from that work.

\begin{definition}
A stabilizer code $S$ is called local-dimension-invariant (LDI) iff:
\begin{equation}
    \phi_\infty(s_i)\odot \phi_\infty(s_j)=0,\quad \forall s_i,s_j\in S.
\end{equation}
\end{definition}

As an example, consider the two qubit code generated by $\langle X\otimes X,Z\otimes Z\rangle$. The symplectic product between the two generators is $2$, so it makes it a valid qubit code, however, $2\mod p\neq 0$ unless $p=2$, so it is not a valid qudit code for $p\neq 2$. If we instead transform the code into one generated by $\langle X\otimes X^{-1},Z\otimes Z\rangle$, then the symplectic product is now $0$, and so it can be used as generators for any choice of local-dimension, and so is an LDI code. The next statement explains that it is always possible to do so \cite{gunderman2020local}:

\begin{theorem}\label{inv}
All stabilizer codes, $S$, can be put into an LDI form. One such method is to put $S$ into canonical form $[I_k\ X_2\ |\ Z_1\ Z_2]$ then transform the code into $[I_k\ X_2\ |\ Z_1+L\ \ Z_2]$, with $L_{ij}=\phi_\infty(s_i)\odot \phi_\infty(s_j)$ when $i>j$ and $0$ otherwise.
\end{theorem}
\if{false}
\begin{proof}
Begin by putting $S$ into canonical form so that:
\begin{equation}
    \phi_q(S)=[I_{n-k}\ \ X_2\ |\ Z_1\ Z_2],
\end{equation}
then the modified code given by:
\begin{equation}
    \phi_\infty(S'):=[I_{n-k}\ \ X_2\ |\ Z_1+L\ \ Z_2],
\end{equation}
where $L_{ij}=\phi_\infty (s_i)\odot \phi_\infty(s_j)$ when $i>j$ and $0$ otherwise. Then $S'$ satisfies the LDI condition and has $\phi_q(S')=\phi_q(S)$.
\end{proof}
\fi
Note that this does not say all codes have a \textit{unique} LDI form, just that there exists one. The proof used is useful as it gives a prescriptive method for turning a code into an LDI form, however, if one does not put the code into canonical form, the code can still be transformed into an LDI form as this process is equivalent to finding solutions to an integer linear program with an abundance of variables. As the code is put into canonical form in this prescriptive method, we know that the rank of the matrix will be preserved by this operation. All LDI forms ought to also preserve the rank, or equivalently, the number of independent generators. 

As of this point we have merely generated a set of commuting operators that are local-dimension independent. This does not provide for any claims on the distance of the code produced through this method aside from promising that the procedure does not change the distance of the code over the initial local-dimension $q$. For this, we have the following Theorem:
\begin{theorem}
For all primes $p>p^*$, with $p^*$ a cutoff value greater than $q$, the distance of an LDI form of a non-degenerate stabilizer code $[[n,k,d]]_q$ applied over $p$ bases, $[[n,k,d']]_p$, has $d'\geq d$.
\end{theorem}
There are two caveats to this result, one of which we resolve here, the other of which we provide an improvement on. Let $B$ be the maximal entry in $\phi_\infty(S)$. Firstly, this result is only for the case of non-degenerate codes. We will resolve this with Theorem \ref{degen}. Secondly, the initially proven bound was $p^*=B^{2(d-1)}(2(d-1))^{(d-1)}$, which grows very rapidly. While it was true that all primes below $p^*$ could have their distances checked computationally, this still left a large number of primes to check in most cases. In this work we manage to prove an alternative bound that has a nearly quadratic improvement on the dependency on $B$. In the next section we show this alternative cutoff bound, while in the section thereafter the ability to provide a distance promise for degenerate codes is proven and differences between the cases are discussed.


While the proof of Theorem \ref{inv} from \cite{gunderman2020local} used $L_{ij}=\phi_\infty(s_i)\odot \phi_\infty(s_j)$ in order to generate a single LDI form, we may generate other LDI forms by altering the added $L$ matrix. We note two of these now: $L^{(+)}$ and $L^{(-)}$.

\begin{definition}
$L^{(+)}$ ($L^{(-)}$) has $L_{ij}^{(+)}$ ($L_{ij}^{(-)}$) is $\phi_\infty(s_i)\odot \phi_\infty(s_j)$ if the symplectic product is greater than zero (less than zero).
\end{definition}

These alternative $L$ matrices each provide a different property. Firstly, using $L^{(+)}$ allows $\phi_\infty(S)$ to have only non-negative entries. There are certain properties that are only generally true for matrices with non-negative entries, so this can perhaps be of use. Additionally, this could be of use for systems formally with countably infinite local-dimension, such as Bosonic systems, where operators with negative powers are not feasible. Secondly, $L^{(-)}$ permits a slight reduction in the bound for the maximal entry in $\phi_\infty(S)$, as the following Lemma shows:

\begin{lemma}\label{bbound}
The maximal entry in $\phi_\infty(S)$, $B$, can be at most $(1+k(q-1))(q-1)$, and generally $B\leq \max_{i,j}|\phi_\infty(s_i)\odot \phi_\infty(s_j)|$.
\end{lemma}

Upon putting the code into canonical form this follows immediately from the definition of $L^{(-)}$ as each entry will be whatever value was already in that location (values in $\mathbb{Z}_q$) minus the absolute value of the inner product, which will be at most an absolute value of the inner product. While this is a small improvement on the value of $B$, since it's the base of an exponential expression this amounts to a larger improvement in the overall cutoff value.

We will now move to proving an alternative bound on the local-dimension needed in order to promise the distance is at least preserved. The first proof of the cutoff bound for the distance promise for LDI codes used random permutations of the entries in $\phi_\infty$. Here we utilize the structure of the symplectic product as well as that of the partitions of the code in terms of its $X$ component and $Z$ component to obtain an alternative bound for all non-degenerate codes. While this bound is looser when $d$ increases, for small $d$ and large $k$ this bound will typically be roughly quadraticly smaller. In particular we will show:

\begin{theorem}\label{improvedbound}
For all primes $p>p^*$ the distance of an LDI representation of a non-degenerate stabilizer code $[[n,k,d]]_q$ over $p$ bases, $[[n,k,d']]_p$, has $d'\geq d$, where we may use as $p^*$ the value:
\begin{equation}
    (B(q-1)(d-1)(1+(d-1)^2(q-1)^{d-1}(d-2)^{(d-2)/2}))^{d-1},
\end{equation}
with $q$ the initial local-dimension, $d$ the distance of the initial code, and $B$ the maximal entry in the $\phi_\infty$ representation of the code.
\end{theorem}

To make claims about the distance of the code we begin by breaking down the set of undetectable errors into two sets. These definitions highlight the subtle possibility of the distance reducing upon changing the local-dimension.

\begin{definition}
An unavoidable error is an error that commutes with all stabilizers and produces the $\vec{0}$ syndrome over the integers.
\end{definition}

These correspond to undetectable errors that would remain undetectable regardless of the number of bases for the code since they always exactly commute under the symplectic inner product with all stabilizer generators--and thus all members of the stabilizer group. Since these errors are always undetectable we call them unavoidable errors as changing the number of bases would not allow this code to detect this error.

We also define the other possible kind of undetectable error for a given number of bases, which corresponds to the case where some syndromes are multiples of the number of bases:

\begin{definition}
An artifact error is an error that commutes with all stabilizers but produces at least one syndrome that is only zero modulo the base.
\end{definition}

These are named artifact errors as their undetectability is an artifact of the number of bases selected and could become detectable if a different number of bases were used with this code. Each undetectable error is either an unavoidable error or an artifact error. We utilize this fact to show our theorem.

\begin{proof}

Let us begin with a code with local-dimension $q$ and apply it to a system with local-dimension $p$. The errors for the original code are the vectors in the kernel of $\phi_q$ for the code. These errors are either unavoidable errors or are artifact errors. \if{false} We may rearrange the rows and columns so that the stabilizers and registers that generate these entries that are nonzero multiples of $q$ are the upper left $2d\times 2d$ minor, padding with identities if needed. The factor of 2 occurs due to the number of nonzero entries in $\phi_\infty$ being up to double the weight of the Pauli.\fi The stabilizers that generate these multiples of $q$ entries in the syndrome are members of the null space of the minor formed using the corresponding stabilizers.

Now, consider the extension of the code to $p$ bases. Building up the qudit Pauli operators by weight $j$, we consider the minors of the matrix. These minors of size $2j\times 2j$ can have a nontrivial null space in two possible ways:
\begin{itemize}
    \item If the determinant is 0 over the integers then this is either an unavoidable error or an error whose existence did not occur due to the choice of the number of bases.
    \item If the determinant is not 0 over the integers, but takes the value of some multiple of $p$, then it's $0\mod p$ and so a null space exists.
\end{itemize}
Thus we can only introduce artifact errors to decrease the distance. By bounding the determinant by $p^*$, any choice of $p>p^*$ will ensure that the determinant is a unit in $\mathbb{Z}_p$, and hence have a trivial null space since the matrix is invertible.

We next utilize the structure of the symplectic product more heavily in order to reduce the cutoff local-dimension. Note that for a pair of Paulis in the $\phi$ representation, we may write:
\begin{eqnarray}
    \phi(s_1)\odot \phi(s_2)&=&\phi(s_1)\begin{bmatrix} 0 & -I_n\\ I_n & 0 \end{bmatrix} \phi(s_2)^T\\
    &:=&\phi(s_1)g \phi(s_2)^T
\end{eqnarray}
and so we may consider the commutation for the generators with some Pauli $u$ as being given by $\bigoplus_{i=1}^{n-k} (\phi(s_i)g)\phi(u)^T$, where $\bigoplus$ is a direct sum symbol here, indicating that a vector of syndrome values is returned. This removes the distinction between the two components and allows the symplectic product to act like the normal matrix-vector product. Now, notice that for any Pauli weight $j$ operator, we will have up to $j$ nonzero entries in the $X$ component of the $\phi$ representation and up to $j$ nonzero entries in the $Z$ component. This means that up to $j$ columns in each component will be involved in any commutator.

Next, note that to ensure that an artifact error is not induced it suffices to ensure that there is a nontrivial kernel, induced by the local-dimension choice, which is ensured so long as any $2(d-1)\times 2(d-1)$ minor does not have a determinant which is congruent to the local-dimension. This can be promised by requiring the local-dimension to be larger than the largest possible determinant for such a matrix. Since there will be at most $j$ nonzero entries in each component it suffices to consider $j$ columns from each component and subsets of $2j$ rows of this.

From this reduction, we need only ensure that the local-dimension is larger than the largest possible determinant for this $2j\times 2j$ minor. Let us denote this minor by:
\begin{equation}
    \begin{bmatrix}
    X_1 & Z_1\\
    X_2 & Z_2 
    \end{bmatrix},
\end{equation}
where each block has dimensions $j\times j$. The maximal entries are $q-1$ for $X_1$ and $X_2$, whereas for $Z_1$ and $Z_2$ it is bounded by $B$. We now use the block matrix identity:
\begin{equation}
    det\begin{bmatrix}
    X_1 & Z_1\\
    X_2 & Z_2
    \end{bmatrix}=det(X_1)det(Z_2-X_2X_1^{-1}Z_1).
\end{equation}

Since all entries in $X_1$ are integers and the determinant is, by construction, nonzero, the maximal entry in $X_1^{-1}$ will be at most that of the largest cofactor of $X_1$. The largest cofactor, $\tilde{C}$, will be at most $(q-1)^{d-2}(d-2)^{(d-2)/2}$, as provided by Hadamard's inequality. The largest entry in $Z_2-X_2X_1^{-1}Z_1$ is then upper bounded by $B(1+(q-1)\tilde{C}(d-1)^2)$. From here, we may apply Hadamard's inequality for determinants again using the given entry bounds, using that each block has dimensions up to $(d-1)\times (d-1)$, which provides $p^*=(q-1)^{d-1}(d-1)^{d-1}(B(1+(q-1)\tilde{C}(d-1)^2))^{d-1}$, or alternatively expressed in terms of our fundamental variables as
\begin{equation}
 (B(q-1)(d-1)(1+(d-1)^2(q-1)^{d-1}(d-2)^{(d-2)/2}))^{d-1}.
\end{equation}
In the case of $q=2$ this reduces to $(B(d-1)(1+(d-1)^2(d-2)^{(d-2)/2}))^{d-1}$.

Lastly, when $j=d$, we can either encounter an unavoidable error, in which case the distance of the code is $d$ or we could obtain an artifact error, also causing the distance to be $d$. It is possible that neither of these occur at $j=d$, in which case the distance becomes some $d'$ with $d<d'\leq d^*$, with $d^*$ being the distance of the code over the integers. 
\end{proof}

Before concluding this section, we provide a brief comparison of this bound to the original one of $B^{2(d-1)}(2(d-1))^{(d-1)}$. The new bound only depends on $B^{d-1}$ opposed to the original $B^{2(d-1)}$, which as the bound on $B$ depends on $k$ means that for codes, or code families, with larger $k$ values the new bound can provide a tighter expression. Unfortunately, however, this new bound is doubly-exponential in the distance of the code $d$, having a dependency of roughly $d^{d^2}$ opposed to the prior dependency of $d^d$, so if one is attempting to promise the distance of a code with a larger distance, this new bound is likely to be far less tight. In summary, this alternative bound is not per se better, however, since one may simply use whichever of the bounds is tighter this alternative bound may provide a lower requirement for the local-dimension needed in order to ensure that the distance of the code is at least preserved.

\if{false}
\begin{proof}

\if{false}
Sufficient to have:
\begin{equation}
    det(A\oplus B)\leq p
\end{equation}
but note that $det(X_d\oplus Z_d)=det(X_d)det(Z_d)$, which provides a large improvement. Goes from $2d\times 2d$ matrices to $d\times d$ and maximal entry $B$ to $B$ and $q$.

$((q-1)B)^{d-1}(d-1)^{d-1}$.\fi
\end{proof}
\fi


Degenerate codes are a uniquely quantum phenomenon, which suggests that they are a crucial class of quantum error-correcting codes in order to obtain certain properties. For a degenerate quantum error-correcting code we must avoid undetectable errors, but also detectable errors which produce the same syndrome but do not map to the same physical codeword. Any LDPC-like quantum error-correcting code will be degenerate, as, equivalently, a quantum error-correcting code is degenerate if there is some stabilizer group member with lower weight than the distance of the code and by construction one would aim to have a high distance for a quantum LDPC code but still $O(1)$ weight for each generator. We show now that a similar distance promise may be made in the degenerate case as was possible in the non-degenerate case, and remark on what differences exist between the two classes in the local-dimension-invariant framework.

\begin{theorem}\label{degen}
For all primes $p>p^*$ the distance of an LDI representation of a degenerate stabilizer code $[[n,k,d]]_q$ over $p$ bases, $[[n,k,d']]_p$, has $d'\geq d$, where $p^*$ is the same function of $n$, $k$, $d$, and $q$ as before.
\end{theorem}

\begin{proof}
In the case of non-degenerate codes all undetectable errors, up to distance $d$, were in the normalizer of the generators, $\mathcal{N}(S)$, as the weight of all members of the stabilizer group have weight at least $d$. For degenerate codes we only need to be concerned about elements in $\mathcal{N}(S)/S$, as now there are some members of the stabilizer group which might have weight below $d$. The latter set is a subset of the former ($\mathcal{N}(S)/S\subset \mathcal{N}(S)$), and so the same distance promise is obtained as before.
\end{proof}

Notice that all Paulis with weight less than $d$ that are in $S$ produce a syndrome that is all zeros, over the integers, and so may appear to be within the category of unavoidable errors when syndromes are computed. This means that when checking the distance this must carefully be taken into account, otherwise the members in $S$ may be mistaken for these errors leading to an erroneous distance value.


\if{false}
For the degenerate case the concern is having two syndromes map to each other that weren't before. This can be avoided if $p>p_d^*$ where $p_d^*=2d(q-1)^2$ (which is the maximal value for a syndrome value (might be $B^2$ instead)). Not sure about this. Basically detectable errors are the only difference here. Is it always the case that $p_d^*<p^*$? If so then degenerate codes are no different from nondegenerate, in terms of proof. Recast as a difference of two Paulis, with a vector congruent to zero over $p$. The $2d\times 2d$ minors used can be independent so this will replace $B$ with $2B$, but otherwise be the same, I think. This suffices but is not strictly needed.

What if the difference is exactly zero. Is this already covered by the original code $q$? I believe so. Need to argue this one too.

No, this vector difference is not working either since this allows for differing vectors, it's not the case that a trivial null space suffices.

$Ax-By=0\mod p\Longleftrightarrow (A-B\tau)x=0\mod p$, with $\tau$ being some permutation (or restricted permutation) matrix? This permutes the rows and columns so that both matrices act on the same space. I think this patches the above.

\begin{lemma}
Reduction of $p^*$, maybe through linear combination argument--can we get $B$ down to $2q$?
\end{lemma}
\fi

This means that just like non-degenerate quantum codes, we may also promise the distance of the code in the degenerate case, and with the same cutoff bound. While this cutoff value is large, it provides some local-dimension value beyond which the distance will be kept and bounds the set of local-dimension values for which the distance must be manually verified.

This provides information about when the distance of the code must be preserved, however, if we apply a code over $q$ levels to a system with $p<q$ levels, is there some range of values for $p$ whereby we know that the distance must decrease? In the non-degenerate case, we denoted this by $p^{**}$, which was given by: \begin{equation}
    \sqrt{1+{n\choose t}^{1/((n-k)-t)}},\quad t=\left\lfloor\frac{d-1}{2} \right\rfloor.
\end{equation}
Whenever $p<p^{**}$, it must be the case that the distance of the code must decrease. The expression for $p^{**}$ was derived by using the generalized quantum Hamming bound, which holds for all non-degenerate codes, however, for degenerate codes this bound does not always hold. This means that for a general degenerate code we have the following Lemma:
\begin{lemma}
There is no corresponding $p^{**}$ that holds for arbitrary degenerate codes. 
\end{lemma}

While not all degenerate quantum codes obey the generalized quantum Hamming bound, there are certain code families which do \cite{quditgeneral,sarvepalli2010degenerate}. For those code families the exact same expression for $p^{**}$ holds as did for non-degenerate codes.

The non-existence of a $p^{**}$ expression for arbitrary degenerate codes provides an opportunity. Consider a code whose initial local-dimension $q$ is far larger than $2$. In the non-degenerate case this $p^{**}$ provides a local-dimension value below which the distance of the code must decrease, but for degenerate codes the lack of this means that it may be possible to apply the code over a far smaller local-dimension, even local-dimension $2$, and still preserve all of the parameters, and particularly the distance. This suggests that it may be possible to import codes into lower local-dimension values than previously expected.

To ground some of the discussions, we provide some examples next.

\if{false}
\begin{example}
As an example in the $\phi$ representation, let us consider the six qubit code, with parameters $[[6,1,3]]_2$, generated by an extension of the five qubit code. This code, unlike the five qubit code, is a degenerate code, as there is a generator with weight $1$, but the code has distance $3$. The generators for this code can be given by $\{YIZXYI,ZXIZYI,ZIXYZI,IIIIIX,IZZZZI\}$ \cite{Grassl:codetables}. In the $\phi_2$ representation this is given by:
\begin{equation}
\setcounter{MaxMatrixCols}{19}
    \begin{bmatrix}
    1 & 0 & 0 & 1 & 1 & 0 & | & 1 & 0 & 1 & 0 & 1 & 0\\
    0 & 1 & 0 & 0 & 1 & 0 & | & 1 & 0 & 0 & 1 & 1 & 0\\
    0 & 0 & 1 & 1 & 0 & 0 & | & 1 & 0 & 0 & 1 & 1 & 0\\
    0 & 0 & 0 & 0 & 0 & 1 & | & 0 & 0 & 0 & 0 & 0 & 0\\
    0 & 0 & 0 & 0 & 0 & 0 & | & 0 & 1 & 1 & 1 & 1 & 0\\
    \end{bmatrix}
\end{equation}
We perform the following operations to put the code into canonical form: swap rows $(4,5)$, $H$ on register $4$, swap registers $(5,6)$, then add row $4$ to rows $2$ and $3$, resulting in:
\begin{equation}
\setcounter{MaxMatrixCols}{19}
    \begin{bmatrix}
    1 & 0 & 0 & 0 & 0 & 1 & | & 1 & 0 & 1 & 1 & 0 & 1\\
    0 & 1 & 0 & 0 & 0 & 1 & | & 1 & 1 & 1 & 0 & 0 & 0\\
    0 & 0 & 1 & 0 & 0 & 0 & | & 1 & 1 & 1 & 1 & 0 & 0\\
    0 & 0 & 0 & 1 & 0 & 0 & | & 0 & 1 & 1 & 0 & 0 & 1\\
    0 & 0 & 0 & 0 & 1 & 0 & | & 0 & 0 & 0 & 0 & 0 & 0\\
    \end{bmatrix}
\end{equation}
\begin{equation}
    L=\begin{bmatrix}
    0 & 0 & 0 & 0 &0\\
    0 & 0 & 0 & 0 &0\\
    0 & 0 & 0 & 0 &0\\
    0 & -2 & 0 & 0 &0\\
    0 & 0 & 0 & 0 &0
    \end{bmatrix}
\end{equation}
\begin{equation}
\setcounter{MaxMatrixCols}{19}
    \begin{bmatrix}
    1 & 0 & 0 & 0 & 0 & 1 & | & 1 & 0 & 1 & 1 & 0 & 1\\
    0 & 1 & 0 & 0 & 0 & 1 & | & 1 & 1 & 1 & 0 & 0 & 0\\
    0 & 0 & 1 & 0 & 0 & 0 & | & 1 & 1 & 1 & 1 & 0 & 0\\
    0 & 0 & 0 & 1 & 0 & 0 & | & 0 & -1 & 1 & 0 & 0 & 1\\
    0 & 0 & 0 & 0 & 1 & 0 & | & 0 & 0 & 0 & 0 & 0 & 0\\
    \end{bmatrix}
\end{equation}
This has $B=1$, so $p^*=16$. This leaves local-dimensions $3$, $5$, $7$, $11$, and $13$ to have their distances computationally verified, after which we have a $[[6,1,3]]_q$ code for any choice of prime $q$. This degenerate code example is of limited interest, however, as it is derived from the 5-register code, which already is in an LDI form and is known to have the distance always preserved \cite{chau1997correcting}.

\end{example}
\if{false}
\begin{example}
9-qubit shor code, I guess
\end{example}
\fi
\fi

\begin{example}
Consider the qubit code with generators $\langle s_1,s_2\rangle=\langle X^{\otimes n},Z^{\otimes n}\rangle$, with $n\geq 4$ being an even number. This code has parameters $[[n,n-2,2]]_2$. $|\phi_\infty(s_1)\odot \phi_\infty(s_2)|=n$, so directly applying Lemma \ref{bbound}, $B= n-1$ is obtained. This provides a bound of $2(n-1)^2$ via the prior bound, while with the new bound shown here this is $(n-1)$. The results then say that the distance is preserved for $p>n-1$.

Let us take the LDI form for the code as $\langle X^{1-n} X^{\otimes (n-1)},Z^{\otimes n}\rangle$. Observe that all weight one Paulis do not commute with at least one generator for the code, whereas $IZZ^{-1}I^{\otimes (n-3)}$ is an unavoidable error, so the distance is always $d=2$.

While this suggests that the determinant bound we showed is incredibly loose, we can write the qubit code in a different LDI form as $\langle (XX^{-1})^{\otimes (n/2)},Z^{\otimes n}\rangle$. For this form $B=1$, which provides $p^*=2$, using either bound, which means that the distance is always at least preserved. This illustrates the impact of careful selection of the LDI form used, and suggests that perhaps with a careful choice of LDI form the bounds provided can be tight for a given code.
\end{example}

\begin{example}
As another example let us consider the Shor code, with parameters $[[9,1,3]]_2$, and consider a local-dimension-invariant form for it. The Shor code is a degenerate code as the inner blocks of the code have some repeated syndromes. The code has a maximal symplectic product of $2$, meaning that there is an LDI form which has $B=1$. One such option is the following set of generators:
\begin{multline}
    \langle XX^{-1}II^{\otimes 6},\ IXX^{-1}I^{\otimes 6},\ I^{\otimes 3}XX^{-1}II^{\otimes 3},\\ I^{\otimes 3}IXX^{-1} I^{\otimes 3}, I^{\otimes 6}XX^{-1}I,\ I^{\otimes 6}IXX^{-1},\\
    Z^{\otimes 6}I^{\otimes 3},\ I^{\otimes 3}Z^{\otimes 6}\rangle.
\end{multline}
Using $B=1$, the bound from \cite{gunderman2020local} is tighter, which provides $p^*=16$, meaning that so long as the local-dimension is $17$ or larger the distance will be at least $3$. From here, manual checking, for local-dimensions $3$, $5$, $7$, $11$, and $13$, verifies that the distance is always preserved. There already was a 9-register code \cite{chau1997correcting}, however, this contextualizes the result within the local-dimension-invariant framework. For completeness, a set of logical operators for this code is given by $\bar{X}=XX^{-1}XX^{-1}XX^{-1}XX^{-1}X$ and $\bar{Z}=Z^{\otimes 9}$. For the logical operators we only require that they do not commute with each other, but do commute with the generators for the code--here $\bar{X}\odot \bar{Z}\neq 0$.
\end{example}

Lastly, generally for the logical operators of the local-dimension-invariant representation of degenerate code the same argument holds as was given in \cite{gunderman2020local}. With all of these pieces we have an equally complete description of degenerate LDI codes, and their slight differences, as existed for non-degenerate LDI codes.

\begin{center}
\begin{table}
\begin{tabular}{||c| c| c||} 
 \hline
 Code parameters & Bound from \cite{gunderman2020local} & Bound shown here  \\ [0.5ex] 
 \hline\hline
 $[[9,1,3]]_2$ & $256$ & $400$\\
 \hline
 $[[13,7,3]]_2$ & $65536$ & $6400$\\
 \hline
 $[[21,13,3]]_2$ & $614656$ & $19600$ \\
 \hline
 $[[29,19,4]]_2$ & $13824000000$ & $1481544000$ \\
 \hline
 $[[13,7,3]]_3$ & $12960000$ & $4161600$ \\
 \hline
  $[[27,22,3]]_3$ & $1049760000$ & $37454400$ \\
 \hline
 $[[91,85,3]]_3$ & $218889236736$ & $540841536$ \\
 \hline
 $[[25,22,3]]_5$ & $213813760000$ & $31258240000$ \\
 \hline

\end{tabular}
\caption{This table compares the bounds on $p^*$, above which the distance of the code is known to be preserved, for a few example codes. The bound on $B$ is used for the value of $B$. Examples taken from \cite{Grassl:codetables} for the qubit codes and \cite{quditgeneral} for qudit cases.}
\end{table}
\end{center}


The local-dimension-invariant (LDI) representation of stabilizer codes allows these codes to be applied regardless of the local-dimension of the underlying system. When introduced only non-degenerate codes could be written in local-dimension-invariant form and have their distance promised to be at least preserved, once the system had sufficiently many levels. In this work we have shown an alternative bound for how many levels are needed for the distance to be promised. While this bound suffers a severe dependency on the distance of the code, it does provide a nearly quadratic improvement on the dependency of the largest entry in the LDI form of the code, given by $B$. So while this bound is less helpful in some cases than the original bound it can be a tighter bound in others. Of particular note is the situation where one does not need to guarantee the same distance as the original code, but just some lower distance $\delta$ or larger. In this case the value for $B$ does not change, however, everywhere that a $d$ appears in the expressions for $p^*$ may be replaced by $\delta$. In these cases the quadratic improvement on the dependency on $B$ shown here can become particularly advantageous.

Beyond this, this work has shown that the LDI representation's associated distance promise also exists for degenerate quantum codes, using the same argument as before but overlooked, and so completes the application of this technique to both families of standard stabilizer codes. Degenerate codes are of particular appeal since they are not restricted by the generalized quantum Hamming bound and can at times protect more logical particles than permitted by non-degenerate codes for a given distance and number of physical particles.

Unfortunately, the utility of this method is somewhat limited as both bounds on the required local-dimension are quite large, as indicated in Table I, but as seen in the examples this bound can often be significantly reduced through careful construction of the LDI form. In order to improve the practicality of this technique the value for $p^*$ must be significantly decreased. One way to reduce these bounds is to reduce the expression for $B$, the maximal entry in the LDI representation. To do so, other analysis techniques will be needed beyond simple counting arguments. Since the LDI form for a code is not unique, one possible method may be to solve systems of homogeneous linear diophantine equations, which given the surplus of variables (additions to entries) compared to variables (requirement of commutations to be zero) is likely to yield far smaller bounds on $B$. A starting point for this might include the following works: \cite{griffiths1946note,givens1947parametric}.

The results shown here extend the utility of local-dimension-invariant stabilizer codes, and so naturally there are questions as to what other uses this technique will have. Is it possible to apply this technique to show some foundational aspect of quantum measurements? Can this technique in some way be used for other varieties of stabilizer like codes, such as Entanglement-Assisted Quantum Error-Correcting Codes \cite{brun2006correcting,wilde2008optimal}? If this method can be applied in this situation it is possible that it could remove the need for entanglement use in these codes, so long as the local-dimension is altered. However, even still, the local-dimension required would likely be quite large so the importance of decreasing the bounds for $p^*$ would become that much more.


\section*{Acknowledgments}

We thank Andrew Jena and David Cory for helpful comments, as well as an anonymous referee for comments that simplified, and tightened, the degenerate distance promise proof.

\section*{Funding}
This work was supported by Industry Canada, the Canada
First Research Excellence Fund (CFREF), the Canadian Excellence Research Chairs (CERC 215284) program, the Natural Sciences and Engineering Research Council of Canada
(NSERC RGPIN-418579) Discovery program, the Canadian
Institute for Advanced Research (CIFAR), and the Province
of Ontario.

\if{false}

&\section{Junk for this, but maybe useful for the future}


The problem can be stated as:
\begin{eqnarray}
    \min_{\{B_i\}} \max_{i,j} |B_i|_2^2 &&\\
    (\phi_q(s_i)+B_iq)\odot (\phi_q(s_j)+B_jq)&=&0\\
    B_i\in \mathbb{Z}^{2n}
\end{eqnarray}

\begin{definition}
$L^+=\phi(s_i)\odot \phi(s_j)$ pick the elements that are positive for $i<j$ and otherwise pick $j>i$. This provides the determinant theorem I want. That bound doesn't actually help...
\end{definition}

    $q$, $B-q^2d^2B$ gives $q^dd^d(B(q^2d^2+1))^d$ with ratio:
    
    \begin{eqnarray}
    & &\frac{(q-1)^{d-1}(d-1)^{d-1}(B((q-1)^2d^2+1))^{d-1}}{B^{2(d-1)}(2(d-1))^{(d-1)}}\\
    &=&\frac{(q-1)^{d-1}(((q-1)^2d^2+1))^{d-1}}{B^{(d-1)}2^{(d-1)}}\\
    &=&\left(\frac{(q-1)^2d^2+1}{2(2+k(q-1))}\right)^{d-1}
    \end{eqnarray}
this isn't always smaller than one, but often times is.

$qBd$ $q^3Bd^3+qBd$ gives $(qBd)^d(Bq^3d^3)^d$ waaay worse

\begin{proof}

Let us begin with a code over $q$ bases and extend it to $p$ bases. The errors for the original code are the vectors in the kernel of $\phi_q$ for the code. These errors are either unavoidable errors or are artifact errors. \if{false} We may rearrange the rows and columns so that the stabilizers and registers that generate these entries that are nonzero multiples of $q$ are the upper left $2d\times 2d$ minor, padding with identities if needed. The factor of 2 occurs due to the number of nonzero entries in $\phi_\infty$ being up to double the weight of the Pauli.\fi The stabilizer(s) that generate these multiples of $q$ entries in the syndrome are members of the null space of the minor formed using the corresponding stabilizer(s).

Now, consider the extension of the code to $p$ bases. Building up the qudit Pauli operators by weight $j$, we consider the minors of the matrix. These minors of size $2j\times 2j$ can have a nontrivial null space in two possible ways:
\begin{itemize}
    \item If the determinant is 0 over the integers then this is either an unavoidable error or an error whose existence did not occur due to the choice of the number of bases.
    \item If the determinant is not 0 over the integers, but takes the value of some multiple of $p$, then it's $0\mod p$ and so a null space exists.
\end{itemize}
Thus we can only introduce artifact errors to decrease the distance. By bounding the determinant by $p^*$, any choice of $p>p^*$ will ensure that the determinant is a unit in $\mathbb{Z}_p$, and hence have a trivial null space since the matrix is invertible.

We utilize the structure of the symplectic product more heavily in order to reduce the cutoff dimension. Note that for a pair of Paulis in the $\phi$ representation, we may write:
\begin{eqnarray}
    \phi(s_1)\odot \phi(s_2)&=&\phi(s_1)\begin{bmatrix} 0 & -I_n\\ I_n & 0 \end{bmatrix} \phi(s_2)^T\\
    &:=&\phi(s_1)g \phi(s_2)^T
\end{eqnarray}
and so we may consider the commutation for the generators with some Pauli $u$ as being given by $\bigoplus_{i=1}^{n-k} (\phi(s_i)g)\phi(u)^T$. Now, notice that for any Pauli weight $j$ operator, we will have up to $j$ nonzero entries in the $X$ component of the $\phi$ representation and up to $j$ nonzero entries in the $Z$ component. This means that up to $j$ columns in each component will be involved in any commutator, and so we may consider these $j$ columns alone for this operator and the generators will still be a direct sum: $[X_j\ |\ Z_j]g=[X_j\bigoplus Z_j]$.

Next note that to ensure that an artifact error is not induced it suffices to ensure that there is a nontrivial kernel, induced by the local-dimension choice, for $[X_j\bigoplus Z_j]$, which is ensured by (Kramers/Sylvesters rule) so long as any $2(d-1)\times 2(d-1)$ minor does not have a determinant which is congruent to the local-dimension. This can be promised by requiring the local dimension to be larger than the largest possible determinant for such a matrix: $p> det(X_j\bigoplus Z_j)$.

Through the partition theorem, we may write $X_j\bigoplus Z_j=(X_j\bigoplus I)(I\bigoplus Z_j)$ and so:
\begin{eqnarray}
|det([X_j\ |\ Z_j]g)|&=&|det(X_j\bigoplus Z_j)|\\
&=&|det(X_j)||det(Z_j)|.
\end{eqnarray} From here, we may apply Hadamard's inequality for determinants using the fact that the maximal entries are $q-1$ for $X_j$ and $B$ for $Z_j$, and each has dimension up to $(d-1)\times (d-1)$, which provides:
\begin{equation}
    p^*=((q-1)B)^{d-1}(d-1)^{d-1}
\end{equation}

\if{false}
Now, in order to guarantee that the value of $p$ is at least as large as the determinant, we can use Hadamard's inequality to obtain:
\begin{equation}
    p> p^* =B^{2(d-1)}(2(d-1))^{(d-1)}
\end{equation}
where $B$ is the maximal entry in $\phi_\infty$. Since we only need to ensure that the artifact induced null space is trivial for Paulis with weight less than $d$, we used this identity with $2(d-1)\times 2(d-1)$ matrices.
\fi

Lastly, when $j=d$, we can either encounter an unavoidable error, in which case the distance of the code is $d$ or we could obtain an artifact error, also causing the distance to be $d$. It is possible that neither of these occur at $j=d$, in which case the distance becomes some $d'$ with $d<d'\leq d^*$. 
\end{proof}

$Bd$ in Hadamard's inequality from $L^{(+)}$ since both halves are positive matrices provides $(Bd)^dd^{d/2}$.

\begin{theorem}\label{degen}
The distance promise can also be made for degenerate codes.
\end{theorem}

\begin{proof}
For the undetectable error case, this follows the same reasoning as the non-degenerate case of this theorem, so we only need to worry about two errors with the same syndrome mapping to different logical states.

For the degenerate syndrome error case, let $v$ and $u$ be two Pauli operators, with the same syndrome values, given by $\mathcal{J}$, with weight at most $d$. Let $\mathcal{A}$ be the (at most) $(d-1)$ registers such that $\mathcal{A}\cdot \phi(v)=\mathcal{J}$. Let $\mathcal{B}$ be the (at most) $(d-1)$ registers such that $\mathcal{B}\cdot \phi(u)=\mathcal{J}$. Generally $\phi(v)$ and $\phi(u)$ represent at least partly disjoint registers, however, if we apply a permutation $\tau$ on the registers in $u$ we may make the register locations seem to match: $\mathcal{B}\cdot \phi(u)=(\mathcal{B}\tau^{-1})\cdot (\tau\phi(u))$. Through this permutation $\tau\phi(u)$ and $\phi(v)$ represent the same registers. It suffices then to just prevent any linear combination of the columns in each of the $X$ and $Z$ components of $\mathcal{A}$ and $\mathcal{B}$ from having a non-trivial kernel due to the change in the local-dimension. We then require that:
\begin{equation}
    det(\mathcal{A}-\mathcal{B})<p,
\end{equation}
with $p$ being the new local-dimension. We may bound the determinant as before, although the maximal entries have been double, providing $p_D^*=(4(q-1)B(d-1))^{(d-1)}$.

While that suffices, it is not quite necessary. In many cases this value can be reduced. We may avoid introducing new degenerate syndrome values by ensuring that the local-dimension is larger than the largest possible syndrome value (over the integers), and so we also have $p_D^*=(d-1)(B+(q-1))$ as a valid bound. We may take the tighter of these two as our requirement to ensure the degeneracies do not cause an error which reduces the distance, meaning:
\begin{multline}p_D^*=\min\{(4(q-1)B(d-1))^{(d-1)},\\ (d-1)(B+(q-1))\}.\end{multline}
As $4B(q-1)\geq B+(q-1)$, the latter of these is always at most the size of the former, meaning that $p_D^*=(d-1)(B+(q-1))$ suffices.

In totality, we must have $p>\max\{p^*,p_D^*\}$ in order to ensure the distance is preserved.

\if{false}
It suffices, although is not necessary, to avoid allowing $(A-B\tau^{-1})$ to have a nontrivial nullspace that is introduced by changing the local-dimension. As in the undetectable error case, we may reduce this to bounding the determinant. In this case the maximal entry is up to twice as large, so we obtain $p_d^*=$.

For the degenerate syndrome error case, let $v$ and $u$ be two Pauli operators, with the same syndrome values, given by $\mathcal{J}$, with weight at most $d$. Let $A$ be the (at most) $(2d)\times (2d)$ matrix minor such that $A\cdot \phi(v)=\mathcal{J}$. Let $B$ be the (at most) $(2d)\times (2d)$ matrix minor in $C$ such that $B\cdot \phi(u)=\mathcal{J}$. $\phi(v)$ and $\phi(u)$ represent at least partly disjoint registers, however, if we apply a permutation $\tau$ on the registers in $u$ we may make the registers seem to match: $B\cdot \phi(u)=(B\tau^{-1})\cdot (\tau\phi(u))$. Through this permutation $\tau\phi(u)$ and $\phi(v)$ represent the same registers and have the same values (upon applying a scaling to the columns via a Clifford operation)--it suffices to just prevent any column combos. We then have that:
\begin{equation}
    (A-B\tau^{-1})\phi(v)=0.
\end{equation}
It suffices, although is not necessary, to avoid allowing $(A-B\tau^{-1})$ to have a nontrivial nullspace that is introduced by changing the local-dimension. As in the undetectable error case, we may reduce this to bounding the determinant. In this case the maximal entry is up to twice as large, so we obtain $p_d^*=$.

Also $p>(n-k-1)B+(n+k+1)(q-1)$ would suffice since the syndromes can only be that large at most.
\fi
\end{proof}

\begin{remark}
It's perfectly possible to turn a degenerate QECC into a non-degenerate one. Is the reverse possible? Yes, but avoided due to $p^*$ requirement, I think.
\end{remark}

\begin{corollary}
Given the bounds proven here, for all $q$ when $d\geq 3$ and for $q>2$ when $d\leq 2$ the same cutoff value is sufficient.
\end{corollary}

\begin{proof}
Consider $d=1$, in which case $p_D^*=p^*=0$.

When $d=2$, in order for the same cutoff value to suffice, we would need:
\begin{equation}
    (q-1)B\geq B+(q-1),
\end{equation}
which only holds if $q>2$.

For $d=3$ we have:
\begin{eqnarray}
    p^*&=&(d-1)^2(q-1)^2B^2\\
    &\geq& (d-1)[B^2+(q-1)^2+(d-3)(B+(q-1))]\\
    &\geq& (d-1)(B+(q-1))\\
    &=& p_D^*
\end{eqnarray}
Since this held for $d=3$, by monotonicity of exponentiation by a power greater than $1$ this will hold for all $d\geq 3$.

\end{proof}
\fi
\fi
\phantomsection  
\renewcommand*{\bibname}{References}

\printbibliography

\end{document}